\documentclass[sn-mathphys-num]{sn-jnl}

\usepackage{graphicx}%
\usepackage{multirow}%
\usepackage{amsmath,amssymb,amsfonts}%
\usepackage{amsthm}%
\usepackage{mathrsfs}%
\usepackage[title]{appendix}%
\usepackage{xcolor}%
\usepackage{textcomp}%
\usepackage{manyfoot}%
\usepackage{booktabs}%
\usepackage{algorithm}%
\usepackage{algorithmicx}%
\usepackage{algpseudocode}%
\usepackage{listings}%
\usepackage{subfig}%
\usepackage{comment}%
\usepackage{xcolor}%

\theoremstyle{thmstyleone}%
\newcommand\beq{\begin{equation}}%
\newcommand\eeq{\end{equation}}%
\newtheorem{prop}{Proposition}%
\newtheorem{cor}{Corollary}%

\begin{document}

\title[Article Title]{Chaotic dynamics of a continuous and discrete generalized Ziegler pendulum}

\author[1]{\fnm{Stefano} \sur{Disca}}
\author*[2]{\fnm{Vincenzo} \sur{Coscia}}\email{cos@unife.it}

\affil[1]{\orgdiv{Department of Physics and Earth Science}, \orgname{University of Ferrara}, \orgaddress{\street{Via Saragat 1}, \city{Ferrara}, \postcode{44122}, \country{Italy}}}
\affil*[2]{\orgdiv{Department of Mathematics and Computer Science}, \orgname{University of Ferrara}, \orgaddress{\street{Via Machiavelli 30}, \city{Ferrara}, \postcode{44121}, \country{Italy}}}

\abstract{We present analytical and numerical results on integrability and transition to chaotic motion for a generalized Ziegler pendulum, a double pendulum subject to an angular elastic potential and a follower force. Several variants of the original dynamical system, including the presence of gravity and friction, are considered, in order to analyze whether the integrable cases are preserved or not in presence of further external forces, both potential and non-potential. Particular attention is devoted to the presence of dissipative forces, that are analyzed in two different formulations. Furthermore, a study of the discrete version is performed. The analysis of periodic points, that is presented up to period 3, suggests that the discrete map associated to the dynamical system has not dense sets of periodic points, so that the map would not be chaotic in the sense of Devaney for a choice of the parameters that corresponds to a general case of chaotic motion for the original system.}




\maketitle

\section{Introduction}\label{sec_introduction}
The Ziegler pendulum is a classical dynamical system introduced by Ziegler \cite{Ziegler}, in order to study the applicability of the stability criteria to non-conservative systems \cite{Pfluger}. In his treatment, Ziegler studied the motion for small angles of a planar physical double pendulum composed by two rigid rods and subject to gravity, angular elastic potentials on the pins and a follower force along the lower rod. Such as the classical double pendulum \cite{Shinbrot, Stachowiak, Dullin}, the system shows itself to be chaotic. The simpler version considered in \cite{Polekhin}, that is a mathematical double pendulum not subject to gravity, exhibits chaotic motion too, but existence of regular solutions is found in presence of particular symmetries, regardless the initial conditions on the canonical variables of the system. Different treatments of integrability and chaotic dynamics of this system can be found in \cite{Kozlov, Thomsen}. The Ziegler pendulum is well-known for the paradoxical destabilization of the system when subject to damping \cite{Bigoni2011, Kirillov2022, Bigoni2024}, while in general a non-conservative system gains stability in presence of a dissipative force. The role of dissipation in the dynamics of the Ziegler pendulum has been studied, among others, in \cite{D'Annibale}, where it is numerically found the occurrence of stable limit-cycles in presence of damping. For an optimization point of view of damped and undamped system see also \cite{Kirillov2011}.
\begin{figure}
\centering
\includegraphics[width=8cm]{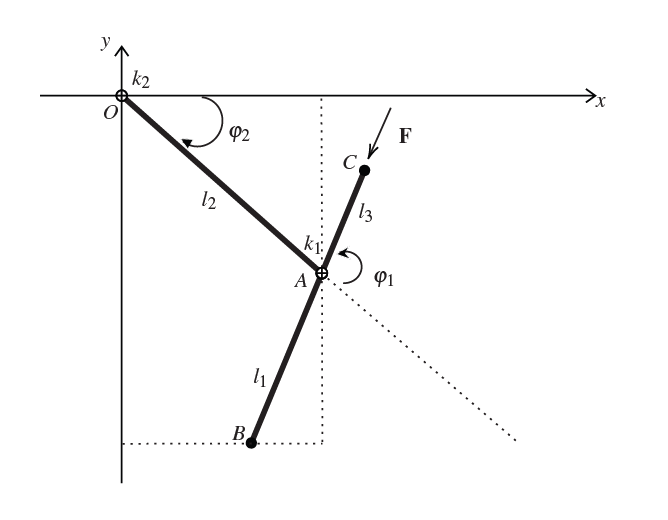}
\caption{A generalized Ziegler pendulum: three material points are joined by two massless rigid rods, with two cylindrical springs located on the pins and a follower force acting on the lower rod.}
\label{ziegler_image}
\end{figure}

In this work we define the Ziegler pendulum as a planar mathematical double pendulum consisting of three material points $A$, $B$, $C$ having mass $m_1$, $m_2$, $m_3$ respectively. The point A is held at constant distance from the point O by means of a massless rod of length $l_2$, while the points $B$ and $C$ are positioned at the ends of a second massless rod hinged in $A$ in such a way that $\overline{BA}=l_1$ and $\overline{AC}=l_3$ (see Figure \ref{ziegler_image}). On the system acts a force of size F always direct as the vector $B-C$ and two elastic forces produced by two cylindrical springs placed on the hinges $O$ and $A$ and having elastic constants $k_1$ and $k_2$. Furthermore, we write $M := m_A + m_B + m_C$ for the total mass and put $\Delta := m_B l_1 - m_C l_3$, that quantifies the distribution of mass on the lower rod.

For the rest of the paper we refer to the above definition of Ziegler pendulum, originally proposed in \cite{Polekhin}. Notice that, by taking $l_1 = 0$ or $l_3 = 0$, this generalized version returns the Ziegler pendulum defined in \cite{Ziegler}, except for the presence of gravity.

By choosing $\varphi_1$, $\varphi_2$ as generalized variables, we have the following expressions for the velocities:
\begin{subequations}
\beq\label{v_A}
v_A^2 = l_2^2 \dot \varphi_2^2
\eeq
\beq\label{v_B}
v_B^2 = l_2^2 \dot \varphi_2^2 + l_1^2 (\dot \varphi_1 + \dot \varphi_2)^2 - 2 l_1 l_2 \dot \varphi_2 (\dot \varphi_1 + \dot \varphi_2) \cos \varphi_1
\eeq
\beq\label{v_C}
v_C^2 = l_2^2 \dot \varphi_2^2 + l_3^2 (\dot \varphi_1 + \dot \varphi_2)^2 + 2 l_2 l_3 \dot \varphi_2 (\dot \varphi_1 + \dot \varphi_2) \cos \varphi_1 \,.
\eeq
\end{subequations}
The elastic potential energy is given by
\beq
\Pi = \frac{1}{2} k_1 \varphi_1^2 + \frac{1}{2} k_2 \varphi_2^2 \,.
\eeq
Furthermore the system is subject to a follower force with constant magnitude
\beq
\textbf{F} = F \frac{\textbf{r}_A - \textbf{r}_C }{|\textbf{r}_A - \textbf{r}_C| } = \begin{pmatrix} - F \cos (\varphi_1 + \varphi_2) \\ F \sin (\varphi_1 + \varphi_2) \end{pmatrix} \,,
\eeq
to which the generalized forces
\begin{subequations}
\beq
Q_1 = \displaystyle \textbf{F} \cdot \frac{\partial \textbf{r}_C}{\partial \varphi_1} = 0
\eeq
\beq
Q_2 = \displaystyle \textbf{F} \cdot \frac{\partial \textbf{r}_C}{\partial \varphi_2} = - F l_2 \sin \varphi_1
\eeq
\end{subequations}
are associated. \\The kinetic energy has the form
\beq
T = \frac{1}{2} \bigg( A_{11} \dot \varphi_1^2 + (A_{12} + A_{21}) \dot \varphi_1 \dot \varphi_2 + A_{22} \dot \varphi_2^2  \bigg) \,,
\eeq
where
\begin{subequations}\label{coeff_A}
\beq
A_{11} = m_B l_1^2 + m_C l_3^2
\eeq
\beq
A_{12} = A_{21} = A_{11} - \Delta l_2 \cos \varphi_1
\eeq
\beq
A_{22} = A_{11} + M l_2^2 - 2 \Delta l_2 \cos \varphi_1
\eeq
\end{subequations}
for the standard Ziegler pendulum. \\The Euler-Lagrange equations can be written as follows:
\begin{subequations}\label{motion2}
\beq
A_{11} \ddot{\varphi}_1 + A_{12} \ddot{\varphi}_2 = r_1
\eeq
\beq
A_{21} \ddot{\varphi}_1 + A_{22} \ddot{\varphi}_2 = r_2 \,,
\eeq
\end{subequations}
where
\begin{subequations}\label{coeff_r}
\beq
r_1 = - k_1 \varphi_1 + \Delta l_2 \dot{\varphi}_2^2 \sin \varphi_1
\eeq
\beq
r_2 = - k_2 \varphi_2 - \Delta l_2 \dot{\varphi}_1 (\dot{\varphi}_1 + 2 \dot{\varphi}_2) \sin \varphi_1 - F l_2 \sin \varphi_1 \,.
\eeq
\end{subequations}
An alternative form for the system \eqref{motion2}, useful also for the numerical calculations, is the following one:
\begin{subequations}\label{motion4}
\beq
\dot{\varphi}_1 = v_1
\eeq
\beq
\dot{v}_1 = \frac{ A_{22} r_1 - A_{12} r_2 }{ A_{11} A_{22} - (A_{12})^2 }
\eeq
\beq
\dot{\varphi}_2 = v_2
\eeq
\beq
\dot{v}_2 = \frac{ A_{11} r_2 - A_{12} r_1 }{ A_{11} A_{22} - (A_{12})^2 } \,.
\eeq
\end{subequations}
It has been proved in \cite{Polekhin} that for $k_2 = 0$ and in presence of two particular symmetries the system is integrable: in the sense of Liouville \cite{Arnold} for the Hamiltonian case $F = 0$, in the sense of Jacobi \cite{Jacobi} for the non-Hamiltonian case $\Delta =  0$. Furthermore, in these two cases there exists a family of periodic solutions that intersect the plane $\varphi_1 = 0$, and the integrability survives for small breaking of these symmetries (i.e. $F$ or $\Delta$ sufficiently small). The integrability of the system results from the cyclicity of the variable $\varphi_2$ in the case $k_2 = 0$; in particular for the Hamiltonian case ($F = 0$) this cyclicity gives rise to a further first integral for the system, that is the conjugate momentum $K = \frac{\partial L}{\partial \dot{\varphi}_2}$.

In view of the results indicated above, in the sequel we analyze the role in the dynamics of further external forces, both conservative and not, in order to verify whether or not the integrabilty of the system survives in presence of them and for $k_2 = 0$. Moreover, we study a discrete map associated to the system \eqref{motion4} in order to check if it verifies the definition of chaoticity in the sense of Devaney \cite{Devaney, Banks} for an arbitrary choice on $k_2$. All the numerical simulations presented have been produced by a Runge-Kutta method of fourth order, with a size step $dt = 0.005$ and a number of iteration $n = 50000$-$500000$, implemented in Python; we also report the computation of the Lyapunov exponents for some cases, by using a proper MATLAB code\footnote{Copyright (c) 2004, Vasiliy Govorukhin. All rights reserved.} that exploits an algorithm proposed in \cite{Wolf} for the calculation of Lyapunov exponents of systems of ordinary differential equations.

The paper is organized as follows. In Section \ref{sec_gravity} we analyze the role of gravity on the dynamics. The effect of linear elastic potentials is treated in Section \ref{sec_springs} in Section \ref{sec_friction}, while a brief analysis of three simple geometric variants.  Then, the role of different kinds of friction on the system is discussed in Section \ref{sec_geom}. Moreover, in Section \ref{sec_discrete} we study a discrete version of the system, where we propose a conjecture on properties of the sets of periodic points. Finally in Section \ref{sec_conclusion} we point out the concluding remarks and reserch perspectives.

\section{Gravity}\label{sec_gravity}
In this section we consider a \emph{vertical} Ziegler pendulum subject to gravity, that is we include in the equations of motion a further potential energy
\beq
V_G = g \sum_{j=A,B,C} m_j y_j = -M g l_2 \sin \varphi_2 + \Delta g \sin (\varphi_1 + \varphi_2).
\eeq
The terms $A_{ij}$ defined in the \eqref{coeff_A} do not change, while for $r_j$ we have
\begin{subequations}
\beq
r_1 = - k_1 \varphi_1 + \Delta l_2 \dot{\varphi}_2^2 \sin \varphi_1 + \Delta g \cos (\varphi_1 + \varphi_2)
\eeq
\beq
r_2 = - \Delta l_2 \dot{\varphi}_1 (\dot{\varphi}_1 + 2 \dot{\varphi}_2) \sin \varphi_1 - F l_2 \sin \varphi_1 - M g l_2 \cos \varphi_2 + \Delta g \cos (\varphi_1 + \varphi_2) \,.
\eeq
\end{subequations}
The integrability of the system is lost for the Hamiltonian case $F = 0$, since the variable $\varphi_2$ is not cyclic anymore. This is not unexpected, since in absence of the external force the system is actually a double pendulum subject to an elastic potential, that does not introduce terms in order to compensate the chaotic behavior.
\begin{figure}
\centering
\includegraphics[width=7cm]{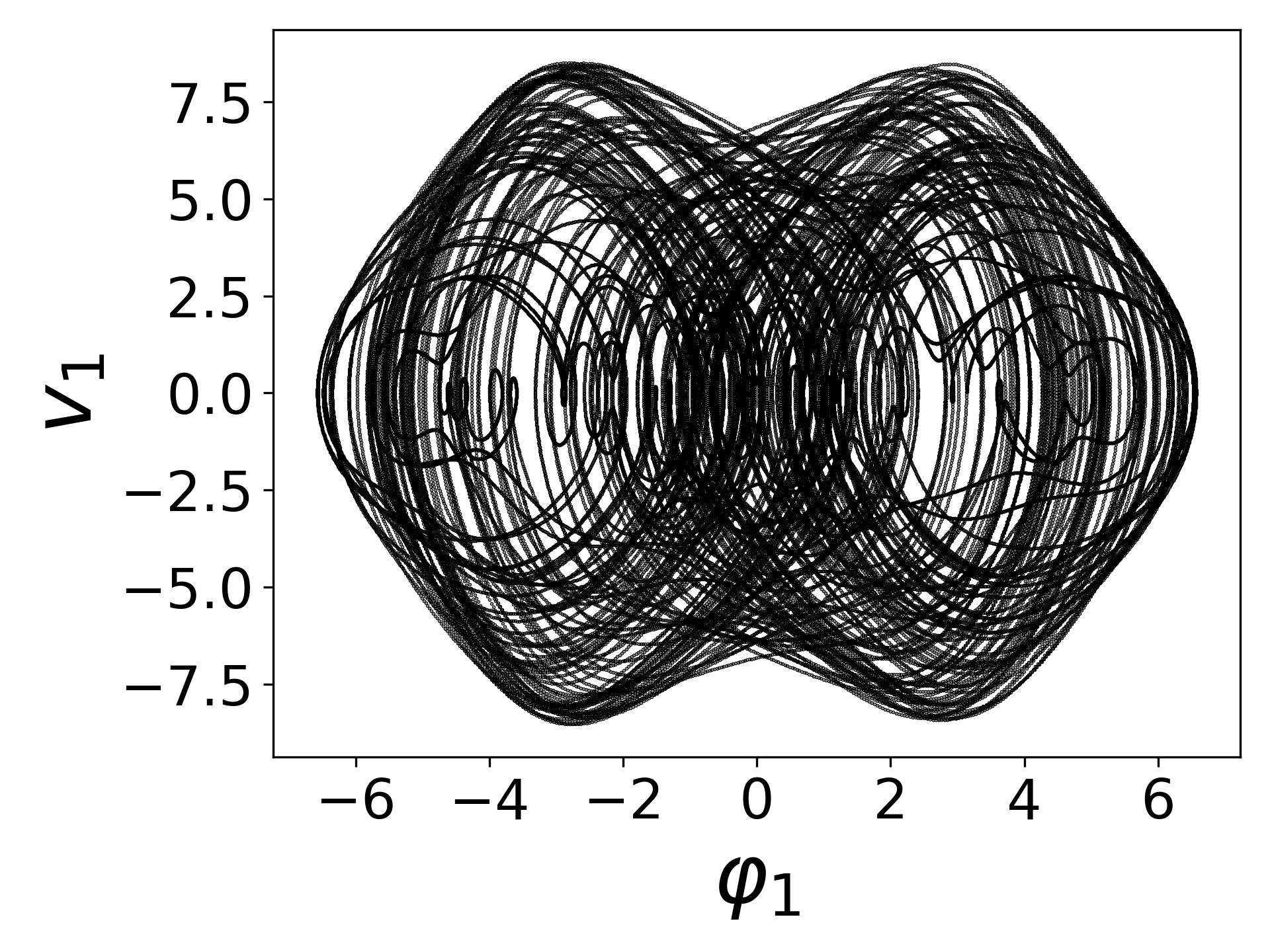}
\caption{Chaotic motion in presence of gravity. Parameters: $m_A = 1$, $m_B = 1.5$, $m_C = 1$,  $l_1 = 1$, $l_2 = 1$, $l_3 = 1$, $k_1 = 2$, $F = 0$. Initial conditions: $\varphi_1(0) = \pi$, $\varphi_2(0) = 0.1$, $v_1(0) = v_2(0) = 0$.}
\end{figure}
\noindent
Anyway, for the non-Hamiltonian case $\Delta = 0$ we have
\begin{subequations}
\beq
A_{11} = A_{12} = A_{21} = m_B l_1^2 + m_C l_3^2
\eeq
\beq
A_{22} = A_{11} + M l_2^2
\eeq
\beq
r_1 = - k_1 \varphi_1
\eeq
\beq
r_2 = - F l_2 \sin \varphi_1 - M g l_2 \cos \varphi_2\,.
\eeq
\end{subequations}
If we think of an angle-dependent form
\beq\label{force_angle}
F(\varphi_1, \varphi_2) = - M g \frac{\cos \varphi_2}{\sin \varphi_1}
\eeq
for the force, the term $r_2$ vanishes and the system \eqref{motion2} reduces to
\begin{subequations}
\beq
A_{11} (\ddot{\varphi}_1 + \ddot{\varphi}_2) = - k_1 \varphi_1
\eeq
\beq
A_{11} \ddot{\varphi}_1 + A_{22} \ddot{\varphi}_2 = 0 \,,
\eeq
\end{subequations}
that is (in this case all coefficients are constant)
\begin{subequations}
\beq\label{force_angle_eq1}
\ddot{\varphi}_1 + \ddot{\varphi}_2 = \frac{- k_1}{A_{11}} \varphi_1
\eeq
\beq\label{force_angle_eq2}
A_{11} \dot{\varphi}_1 + A_{22} \dot{\varphi}_2 = K \,.
\eeq
\end{subequations}
By substituting \eqref{force_angle_eq2} in \eqref{force_angle_eq1} we obtain the following equation independent of $\varphi_2$
\beq
\ddot{\varphi}_1 \bigg( 1 - \frac{A_{11}}{A_{22}} \bigg) + \frac{k_1}{A_{11}} \varphi_1 = 0 \,,
\eeq
that is a one-dimensional harmonic oscillator with frequency
\beq
\omega^2 := \frac{k_1}{A_{11}} \frac{A_{22}}{A_{22} - A_{11} } \,.
\eeq
Given that, the motion for $\varphi_2$ is solved by quadratures
\beq
\varphi_2(t) = \frac{K}{A_{11}} t - \int_{t_0}^t A_{11} \dot{\varphi}_1(t') dt'.
\eeq
This is a very simple result, but it is worth to observe that, in order to obtain it, we must impose a non constant magnitude for the force, that is an \emph{ad hoc} form for the external force that allows $r_2$ to vanish. Furthermore, the expression \eqref{force_angle} is singular if $\varphi_1 = k \pi$ and $\varphi_2 \ne k \pi + \frac{\pi}{2}$ with $k \in \mathbb{N}$, that is for an infinite set of values and initial conditions.

However, even if the gravity destroys the integrable cases known for the Ziegler pendulum, the system shows isolated cases of regular motion for a proper choice of parameters and initial conditions. An interesting feature, that can be found in other chaotic systems \cite{Yanchuk}, is the fast transition shown from a chaotic regime to a regular one. As an example, the motion for two close values of $l_3$ is shown in Figure \ref{gravity_periodic_compare} and \ref{gravity_chaos_compare} (initial conditions and all the other parameters are kept the same). The transition from a regular motion to a chaotic one is confirmed by the computation of the Lyapunov exponents: null Lyapunov exponents are associated to the motion shown in Figure \ref{gravity_periodic_compare}, while the motion shown in Figure \ref{gravity_chaos_compare} exhibits a positive Lyapunov exponent $\lambda_{\varphi_1} \sim 0.19$.
\begin{figure}
\centering
\includegraphics[width=7cm]{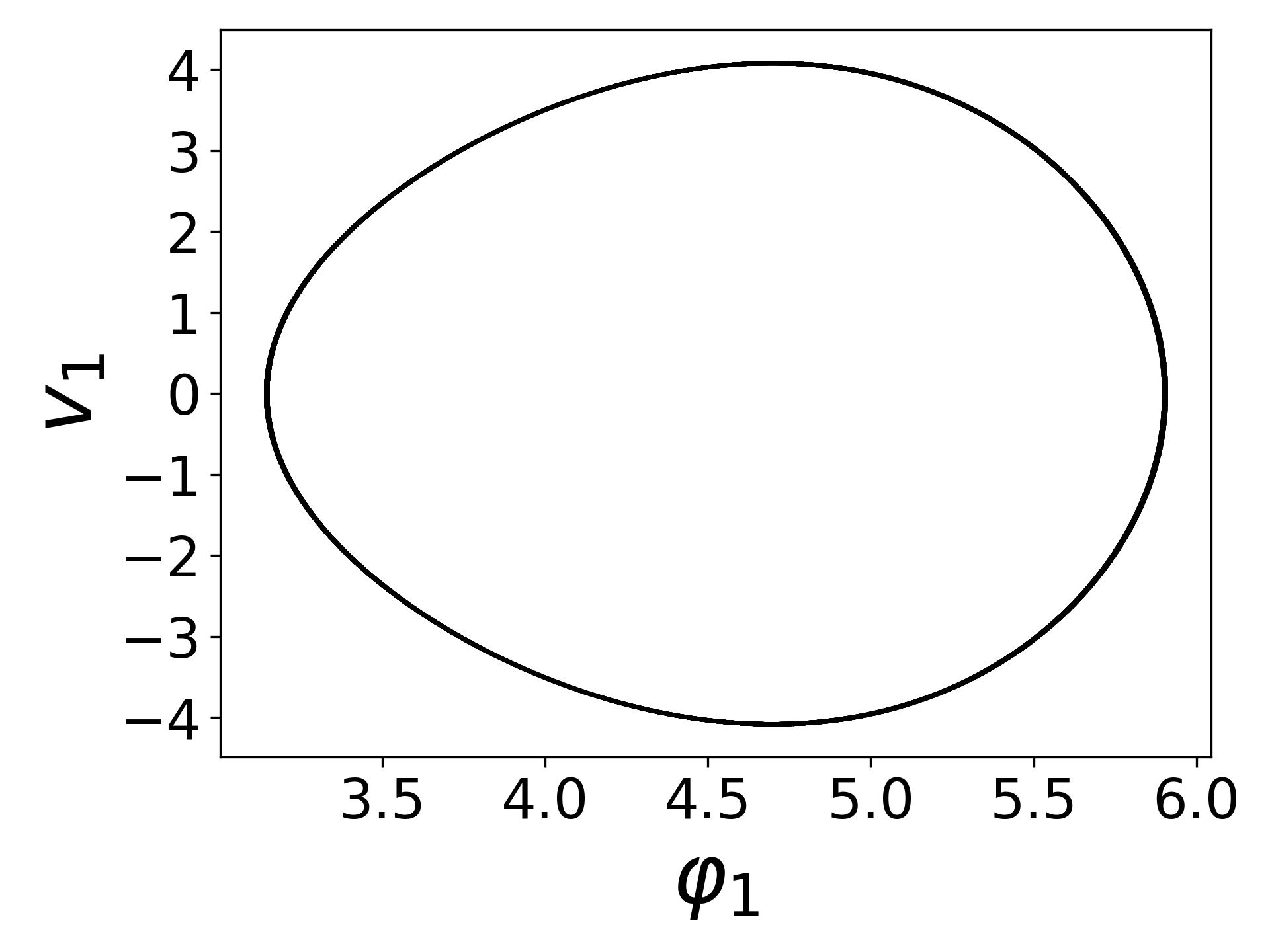}
\caption{Periodic motion in presence of gravity (compare with Figure \ref{gravity_chaos_compare}). Parameters: $m_A = 1$, $m_B = 1$, $m_C = 1$,  $l_1 = 1$, $l_2 = 1$, $l_3 = 1.46$, $k_1 = 1$, $F = 0$. Initial conditions: $\varphi_1(0) = \pi$, $\varphi_2(0) = 0$, $v_1(0) = v_2(0) = 0$.}
\label{gravity_periodic_compare}
\end{figure}
\begin{figure}
\centering
\includegraphics[width=7cm]{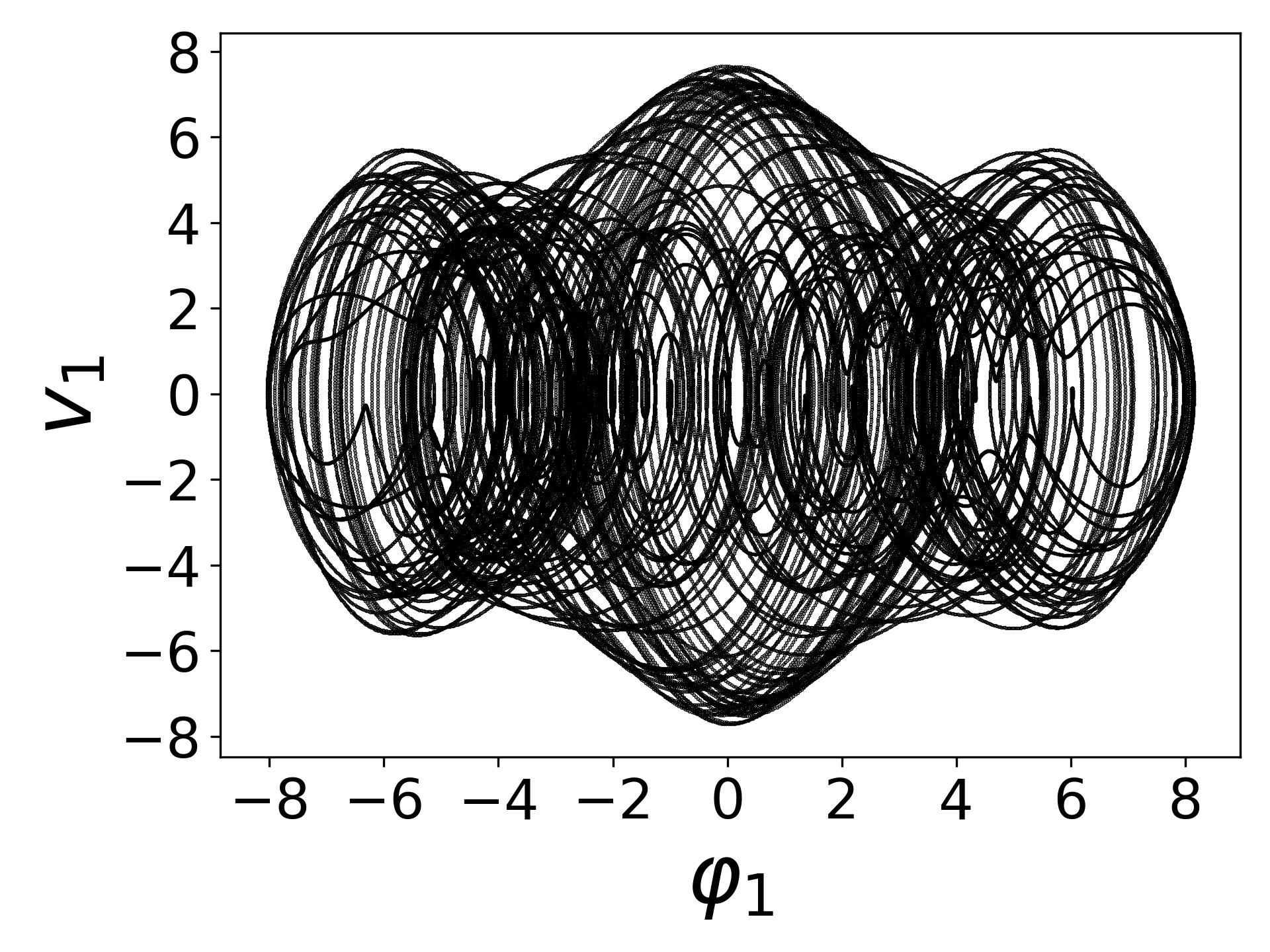}
\caption{Chaotic motion in presence of gravity (compare with Figure \ref{gravity_periodic_compare}). Parameters: $m_A = 1$, $m_B = 1$, $m_C = 1$,  $l_1 = 1$, $l_2 = 1$, $l_3 = 1.45$, $k_1 = 1$, $F = 0$. Initial conditions: $\varphi_1(0) = \pi$, $\varphi_2(0) = 0$, $v_1(0) = v_2(0) = 0$.}
\label{gravity_chaos_compare}
\end{figure}

\section{Linear springs}\label{sec_springs}
In this section we consider Ziegler pendulum with two further linear springs that join the fixed point with the two ends of the lower rods.
\begin{figure}
\centering
\includegraphics[width=8cm]{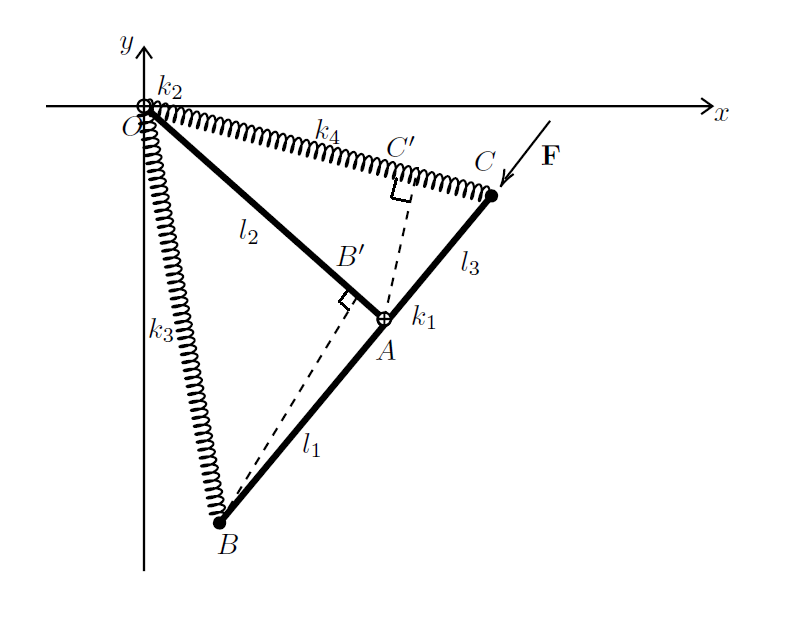}
\caption{The Ziegler pendulum with two linear springs (angles defined in Figure \ref{ziegler_image}).}
\end{figure}
\noindent
The system is now subject to a further elastic potential
\beq
V_{EL} = \frac{1}{2} k_{OB} |\overline{OB}|^2 + \frac{1}{2} k_{OC} |\overline{OC}|^2 \,.
\eeq
Since the elongation of these linear springs depends only on the rotation of the lower rod, as one can figure out from a physical point of view, the elastic potential must be a function of $\varphi_1$; therefore the variable $\varphi_2$ remains cyclic and all the previous results still hold for this model. The quadratic elastic potential is easily evaluated by computing the lengths of the springs. If we denote $B'$ the orthogonal projection of $B$ on the rod $\overline{OA}$, we have
\beq
|\overline{OB}|^2 = l_1^2 + l_2^2 - 2 l_1 l_2 \cos \varphi_1
\eeq
and in analogous way, if we call $A'$ the orthogonal projection of $A$ on the spring $\overline{OC}$, we have
\beq
|\overline{OC}|^2 = l_2^2 + l_3^2 + 2 l_2 l_3 \cos \varphi_1 \,,
\eeq
so that the quadratic elastic potential can be written as
\beq
V_{EL} = l_2 (k_{OC} l_3 - k_{OB} l_1) \cos \varphi_1
\eeq
up to a constant. The equations of motion change as follows: the terms $A_ {ij}$ and $r_2$ still follow the \eqref{coeff_A} and \eqref{coeff_r}, while $r_1$ earns a term due to $\frac{\partial V_{EL}}{\partial \varphi_1}$
\beq
r_1 = - k_1 \varphi_1 + \Delta l_2 \dot{\varphi}_2^2 \sin \varphi_1 - l_2 (k_{OC} l_3 - k_{OB} l_1) \sin \varphi_1 \,.
\eeq
Due to the cyclicity of $\varphi_2$, the cases $F = 0$ and $\Delta = 0$ are still integrable in the sense previously discussed. Furthermore, a new symmetry arises, associated to the case $k_{OB} l_1 = k_{OC} l_3$. This situation physically corresponds to have two external elastic forces that compensate each other, so the motion is the same that we would have for a standard Ziegler pendulum.
\begin{figure}
\centering
\includegraphics[width=7cm]{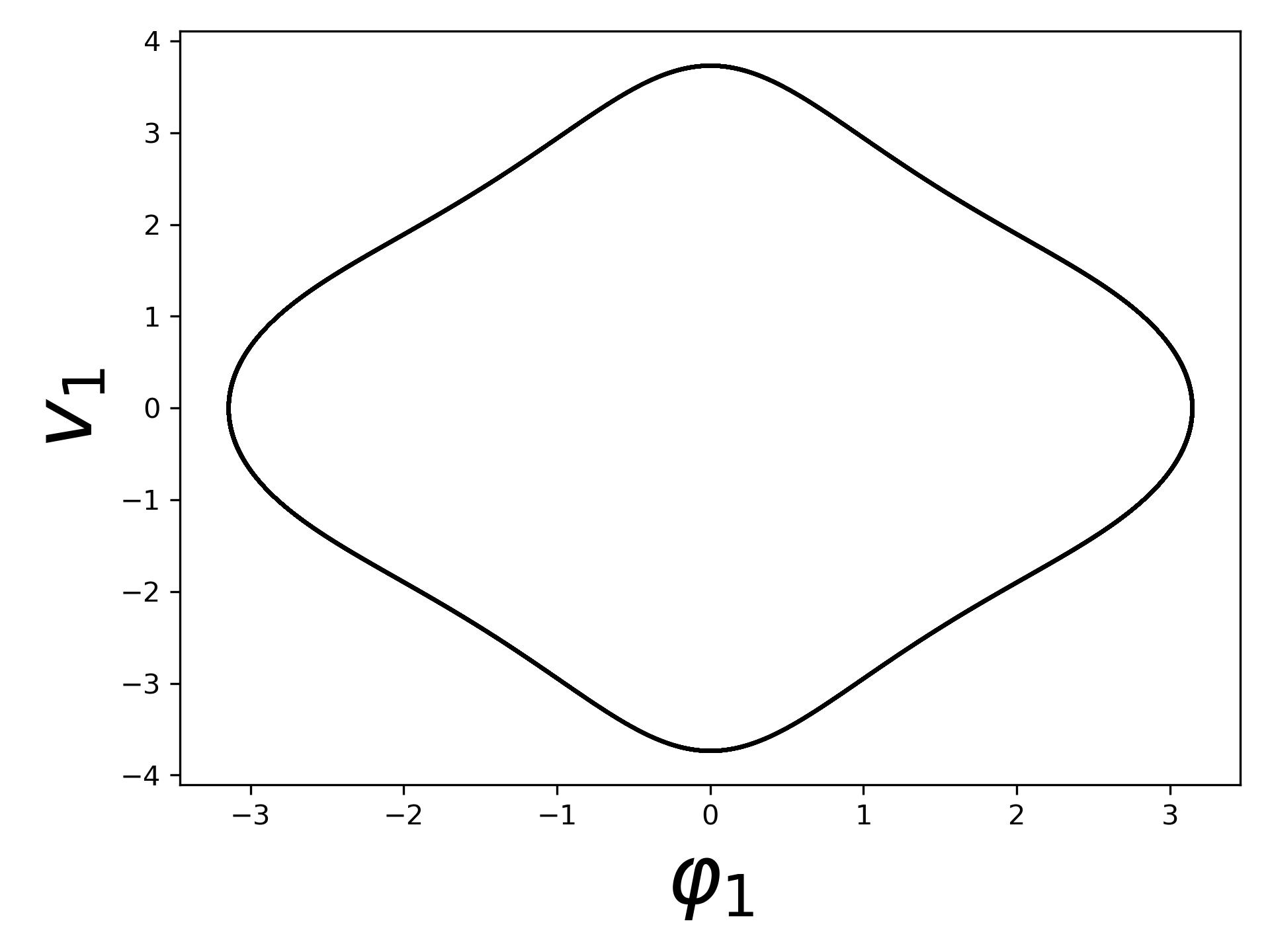}
\caption{Closed trajectory in presence of two linear springs, $F = 0$ and $k_{OB} l_1 = k_{OC} l_3$. Parameters: $m_A = 1$, $m_B = 1$, $m_C = 2$,  $l_1 = 1$, $l_2 = 1$, $l_3 = 2$, $k_1 = 2$, $k_{OB} = 1$, $k_{OC} = 0.5$. Initial conditions: $\varphi_1(0) = \pi$, $\varphi_2(0) = 0.2$, $v_1(0) = 0.1$, $v_2(0) = 0$.}
\end{figure}
\begin{figure}
\centering
\includegraphics[width=7cm]{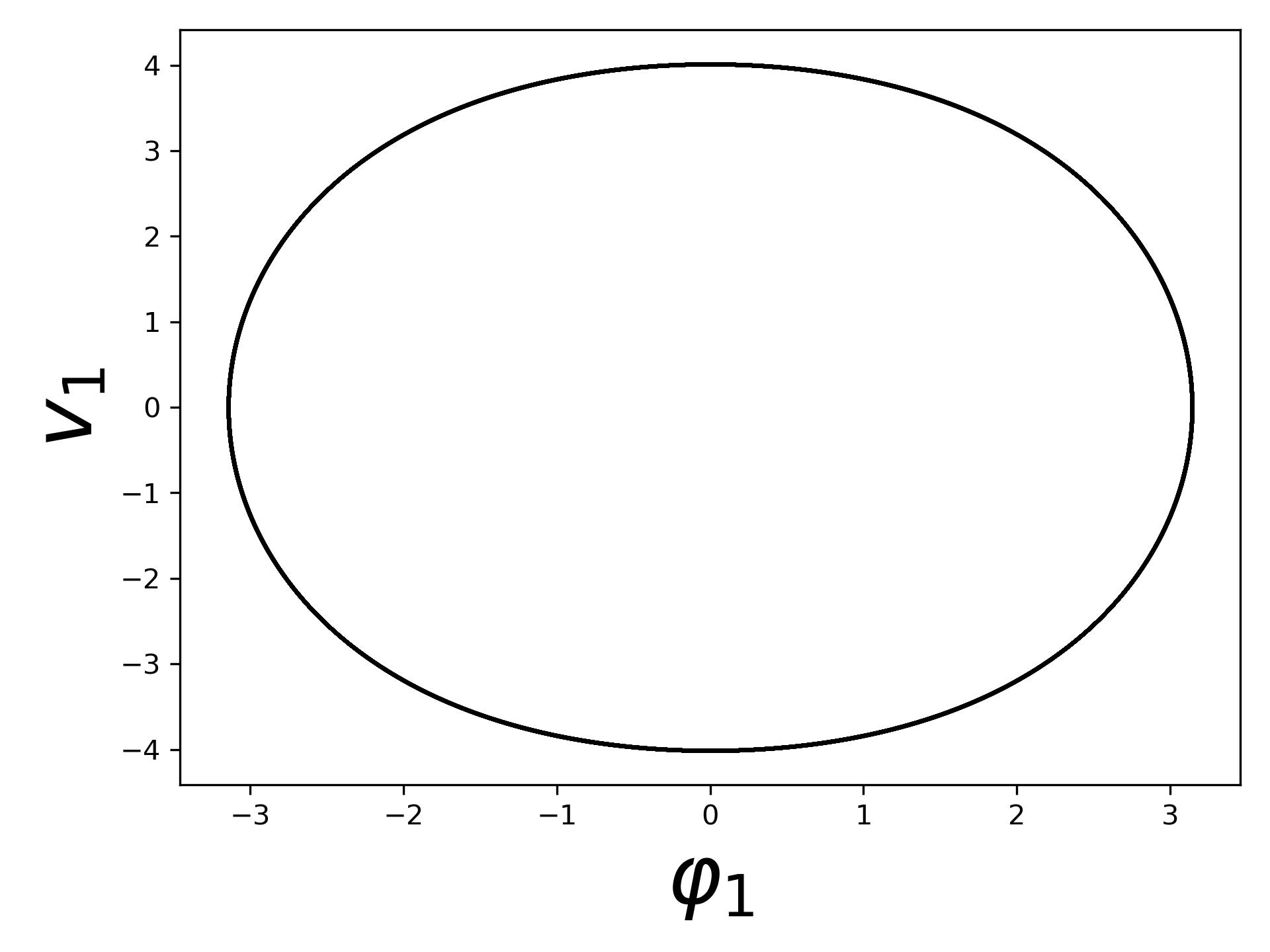}
\caption{Closed trajectory in presence of two linear springs, $\Delta = 0$ and $k_{OB} l_1 = k_{OC} l_3$. Parameters: $m_A = 1$, $m_B = 1$, $m_C = 2$,  $l_1 = 1$, $l_2 = 1$, $l_3 = 0.5$, $k_1 = 2$, $F = 2$, $k_{OB} = 1$, $k_{OC} = 2$. Initial conditions: $\varphi_1(0) = \pi$, $\varphi_2(0) = 0.2$, $v_1(0) = 0.1$, $v_2(0) = 0$.}
\end{figure}
\noindent
However, the system is integrable for all the possible values of $k_{OB}$ and $k_{OC}$, in the sense of Liouville if $F = 0$ or in the sense of Jacobi if $\Delta = 0$. In particular all the results proved in \cite{Polekhin} still hold, such as for example the existence of a family of periodic solutions that intersect the plane $\varphi_1 = 0$. The presence of the linear springs of course modifies the periodic solutions of the system; in particular, they show to have intersections, absent in the standard case, as depicted in Figure \ref{springs_family_F0}.
\begin{figure}
\centering
\includegraphics[width=7cm]{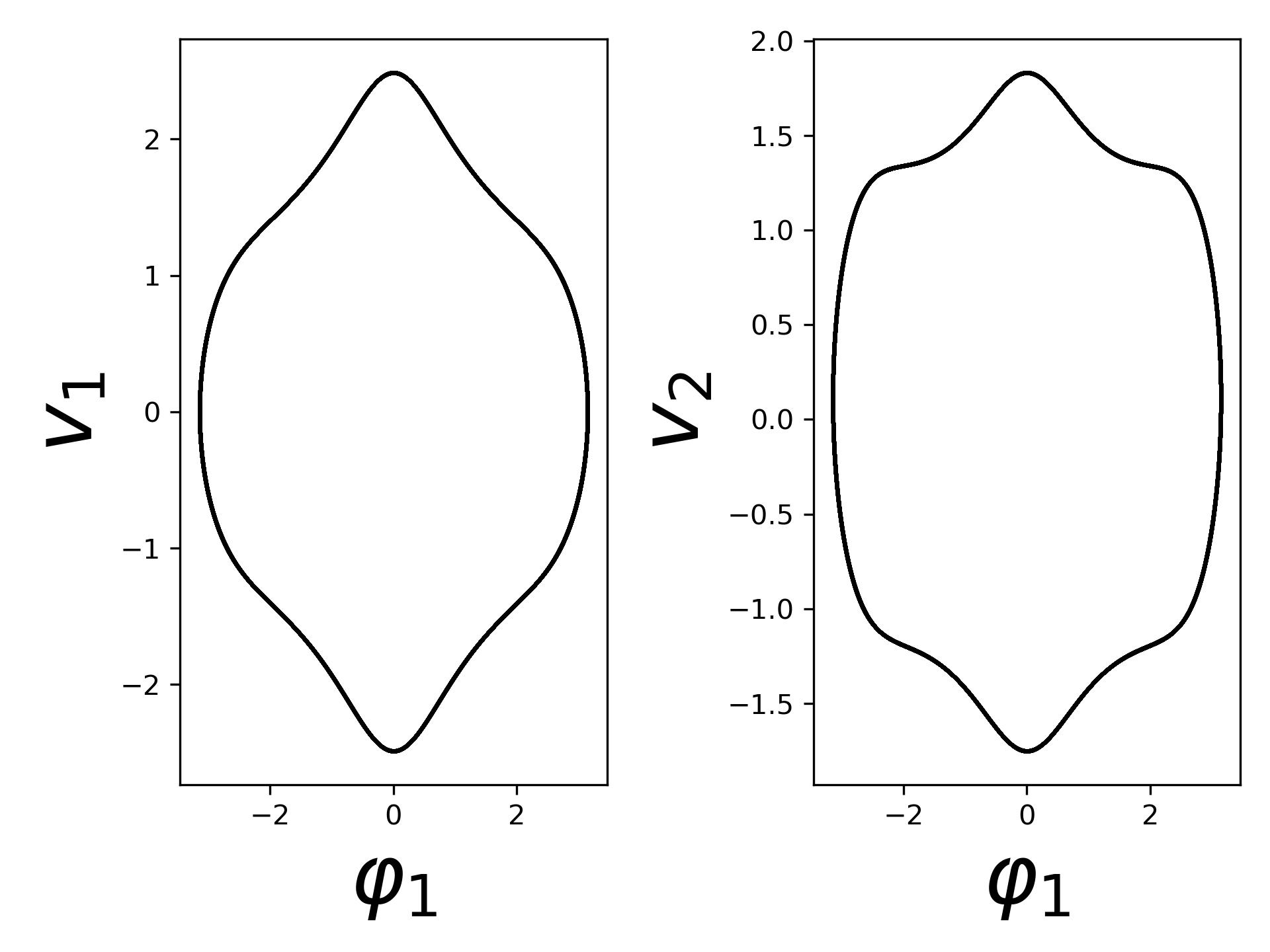}
\caption{Closed trajectories in presence of two linear springs, $F = 0$ and $k_{OB} l_1 \ne k_{OC} l_3$. Parameters: $m_A = 2$, $m_B = 1$, $m_C = 3$,  $l_1 = 1$, $l_2 = 1$, $l_3 = 3$, $k_1 = 2.5$, $k_{OB} = 3$, $k_{OC} = 0$. Initial conditions: $\varphi_1(0) = \pi$, $\varphi_2(0) = 0.2$, $v_1(0) = 0.1$, $v_2(0) = 0$.}
\end{figure}
\begin{figure}
\centering
\includegraphics[width=7cm]{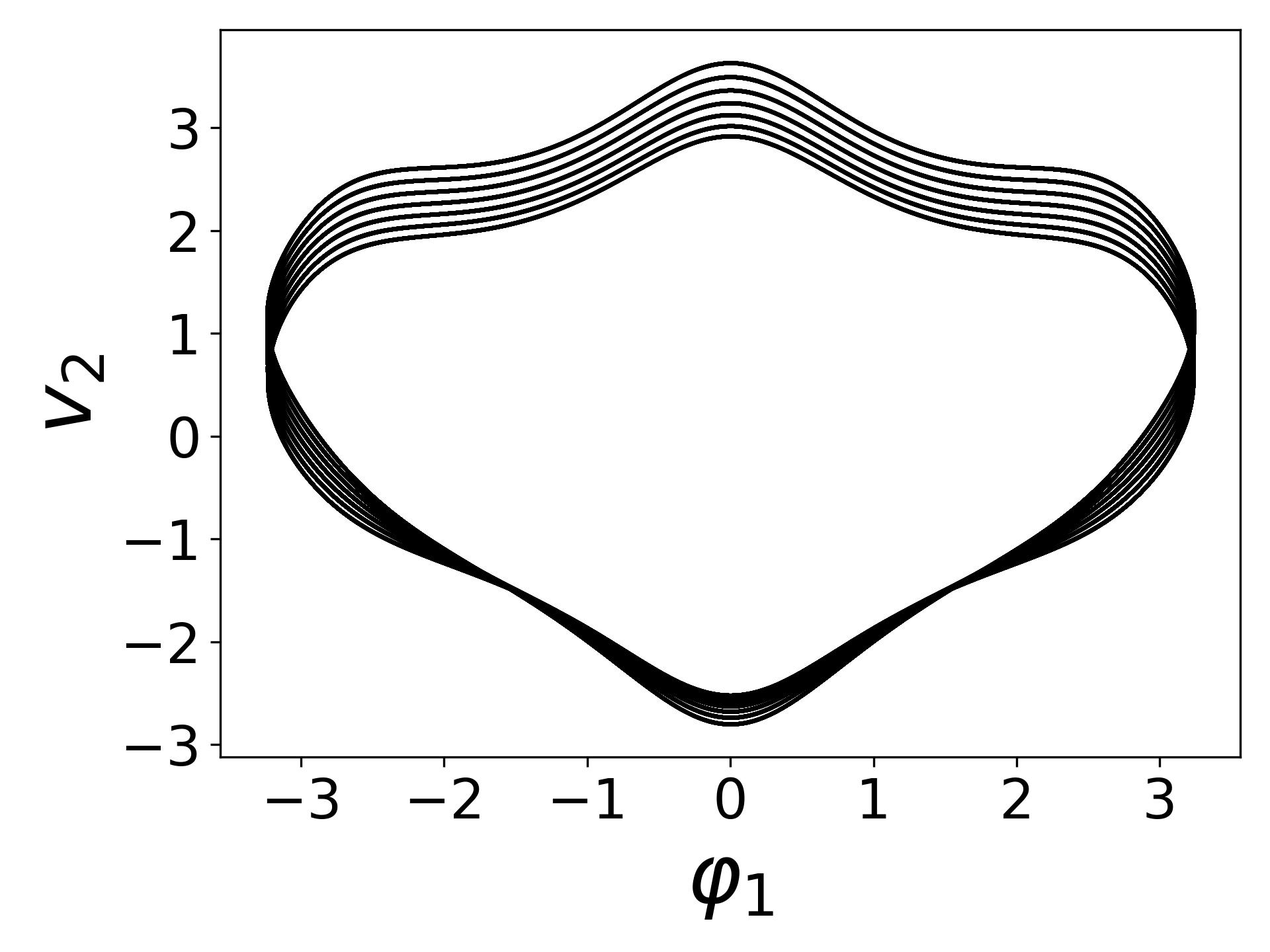}
\caption{A family of periodic solutions in presence of two linear springs and $F = 0$. Parameters: $m_A = 2$, $m_B = 1$, $m_C = 3$,  $l_1 = 1$, $l_2 = 1$, $l_3 = 3$, $k_1 = 2.5$, $k_{OB} = 3$, $k_{OC} = 1$. Initial conditions: $\varphi_1(0) = \pi$, $\varphi_2(0) = \frac{\pi}{2}$, $v_1(0) = 0.5$, $v_2(0) = 0.1 j$ with $j = 0, \dots, 6$.}
\label{springs_family_F0}
\end{figure}

\section{Three geometric variants}\label{sec_geom}
Let us consider a Ziegler pendulum with two physical homogeneous rods, as introduced in \cite{Ziegler} (see Figure \ref{two_phys}). We denote by $C_U = (x_U, y_U)$ and $C_D = (x_D, y_D)$ the centers of mass of the upper and lower pendulum, with masses $m_U$ and $m_D$ respectively; we still use $\varphi_{1,2}$ as generalized coordinates of the system and refer to the lengths previously introduced.
\begin{figure}
\centering
\includegraphics[width=8cm]{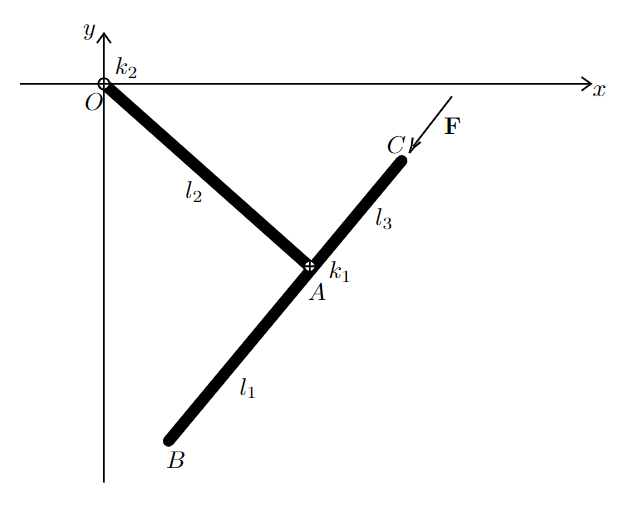}
\caption{A Ziegler pendulum with two physical rods (angles defined in Figure \ref{ziegler_image}).}
\label{two_phys}
\end{figure}
Denoting by $\theta = \varphi_1 + \varphi_2$ the angle between the lower rod and the horizontal axis, it is easy to find
\beq
\begin{cases}
x_U = \frac{l_2}{2} \cos \varphi_2 \\
y_U = - \frac{l_2}{2} \sin \varphi_2
\end{cases}
\eeq
and
\beq
\begin{cases}
x_D = l_2 \cos \varphi_2 - (\frac{l_1 - l_3}{2}) \cos(\varphi_1 + \varphi_2) \\
y_D = - l_2 \sin \varphi_2 + (\frac{l_1 - l_3}{2}) \sin(\varphi_1 + \varphi_2) \,.
\end{cases}
\eeq
The kinetic energy of the upper pendulum, that rotates around the fixed point $O$, is
\beq\label{TU_physical}
T_U = \frac{1}{6} m_U l_2^2 \dot \varphi_2^2 \,,
\eeq
while the kinetic energy of the lower pendulum is found to be
\beq\begin{split}\label{TD_physical}
T_D =& \frac{1}{2} m_D \bigg( l_2^2 \dot \varphi_2^2 + \bigg( \frac{l_1 - l_3}{2} \bigg)^2 (\dot \varphi_1 + \dot \varphi_2)^2 - l_2 (l_1 - l_3) \dot \varphi_2 (\dot \varphi_1 + \dot \varphi_2) \cos \varphi_1 \bigg) + \\
&+ \frac{1}{24} m_D (l_1 + l_3)^2 \dot \varphi_1^2 \,.
\end{split}\eeq
The parameters in the equations of motion \eqref{motion2} are modified in the following way:
\begin{subequations}
\beq
A_{11} = \frac{m_D}{3} (l_1^2 + l_3^3 - l_1 l_3)
\eeq
\beq
A_{12} = A_{21} = \frac{m_D}{4} (l_1 - l_3)^2  - \frac{m_D}{2} l_2 (l_1 - l_3) \cos \varphi_1
\eeq
\beq
A_{22} =  \frac{m_U}{3} l_2^2 + \frac{m_D}{4} \bigg( 4 l_2^2 + (l_1 - l_3)^2 \bigg) - m_D l_2 (l_1 - l_3) \cos \varphi_1
\eeq
\beq
r_1 = - k_1 \varphi_1 + \frac{m_D}{2} l_2 (l_1 - l_3) \dot \varphi_2^2 \sin \varphi_1
\eeq
\beq
r_2 = - k_2 \varphi_2 - \frac{m_D}{2} l_2 (l_1 - l_3) \dot \varphi_1 (\dot \varphi_1 + 2 \dot \varphi_2) \sin \varphi_1 - F l_2 \sin \varphi_1 \,.
\eeq
\end{subequations}
We formally recover the same equations of motion valid for the standard Ziegler pendulum. The symmetry $m_1 l_1 - m_3 l_3 = 0$ is replaced by the symmetry $l_1 - l_3 = 0$, which again corresponds to have the lower center of mass located on the pin between the two rods. The usual two integrable cases hold and the motion is formally the same as for the standard Ziegler pendulum.

For a Ziegler pendulum with a physical upper rod and a mathematical lower rod (see Figure \ref{phys_upper}), we recover the standard Ziegler pendulum with just a constant correction on the coefficient $A_{22}$, that is ($m_U = m_A$)
\beq
A_{22} = A_{11} + \bigg( \frac{1}{3} m_A + m_B + m_C \bigg) l_2^2 - 2 \Delta l_2 \cos \varphi_2 \,.
\eeq
\begin{figure}
\centering
\includegraphics[width=8cm]{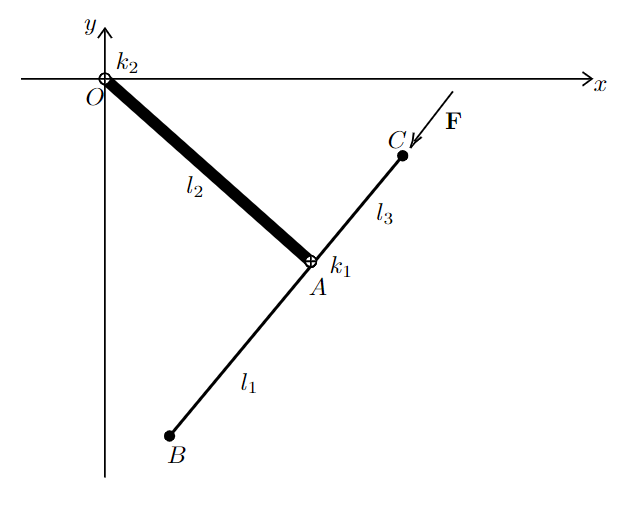}
\caption{A Ziegler pendulum with a physical upper rod and a mathematical lower rod (angles defined in Figure \ref{ziegler_image}).}
\label{phys_upper}
\end{figure}
Conversely, for a Ziegler pendulum with a mathematical upper rod and a physical lower rod (see Figure \ref{phys_lower}), we recover the physical Ziegler pendulum with just a constant correction on the coefficient $A_{22}$, that is ($m_U = m_A$)
\beq
A_{22} =  m_U l_2^2 + \frac{m_D}{4} \bigg( 4 l_2^2 + (l_1 - l_3)^2 \bigg) - m_D l_2 (l_1 - l_3) \cos \varphi_1 \,.
\eeq
\begin{figure}
\centering
\includegraphics[width=8cm]{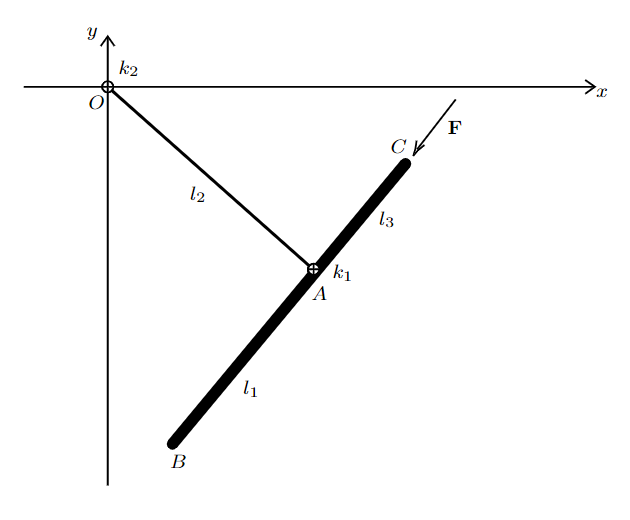}
\caption{A Ziegler pendulum with a mathematical upper rod and a physical lower rod (angles defined in Figure \ref{ziegler_image}).}
\label{phys_lower}
\end{figure}
For both mixed versions the correction does not change the formal shape of the equations of motion and there are no differences in the qualitative motion of the system.

Finally, all these three variants of the Ziegler pendulum (a physical double pendulum and two mixed double pendulums) exhibit the same features of the standard Ziegler pendulum. This is not surprising, since these variants simply change the positions of the centers of mass of the two rods without adding relevant terms in the equations of motion.

\section{Friction}\label{sec_friction}
In this section we analyze the Ziegler pendulum subject to friction in two possible senses: a viscous friction that acts on the three material points and a friction force that acts on the pins of the pendulum being opposite to the rotation.

\subsection{Stokes friction}
Let's consider a friction force that follows
\beq\label{friction_force}
\textbf{A} = - \mu \textbf{v}, \qquad \mu \ge 0 \,,
\eeq
that is, with a proper dimensional choice for $\mu$, the Stokes law for the friction due to a viscous fluid (e.g. the air). From a physical point of view, the presence of three friction coefficients corresponds to having three material points made of three different materials, or even to have the pendulum located in a box with three different separate fluids. We do that in order to consider the most general perturbation to the dynamical system; a realistic model can be obtained for example by equating all the three friction coefficients.

In order to calculate the generalized force associated to the friction \eqref{friction_force}, we consider the following Rayleigh dissipation function
\beq
R := \frac{1}{2} \mu v^2 \implies \textbf{A} = - \nabla_{\textbf{v}} R \,.
\eeq
Since
\beq
\frac{\partial \textbf{v}}{\partial \dot \varphi_j} = \frac{\partial \textbf{r}}{\partial \varphi_j} \,,
\eeq
the components of the generalized force associated to $\textbf{A}$ are given by
\beq
Q_{j} = -\frac{\partial R}{\partial \dot \varphi_j}
\eeq
and the Euler-Lagrange equations become
\beq
\frac{d}{dt} \bigg( \frac{\partial L}{\partial \dot \varphi_j} \bigg) - \frac{\partial L}{\partial \varphi_j} + \frac{\partial R}{\partial \dot \varphi_j} = 0 \,.
\eeq
By recalling \eqref{v_A}, \eqref{v_B} and \eqref{v_C}, the following Rayleigh function and generalized forces are associated to the friction acting on $A$
\beq
R_A = \frac{1}{2} \mu_A l_2^2 \dot \varphi_2^2 \,,
\eeq
\begin{subequations}
\beq
Q_1^A = 0
\eeq
\beq
Q_2^A = -\mu_A l_2^2 \dot \varphi_2\,,
\eeq
\end{subequations}
to the friction acting on $B$
\beq
R_B = \frac{1}{2} \mu_B \bigg( l_2^2 \dot \varphi_2^2 + l_1^2 (\dot \varphi_1 + \dot \varphi_2)^2 - 2 l_1 l_2 \dot \varphi_2 (\dot \varphi_1 + \dot \varphi_2) \cos \varphi_1 \bigg) \,,
\eeq
\begin{subequations}
\beq
Q_1^B = - \mu_B \bigg( l_1^2 (\dot \varphi_1 + \dot \varphi_2) - l_1 l_2 \dot \varphi_2 \cos \varphi_1 \bigg)
\eeq
\beq
Q_2^B = - \mu_B \bigg( l_1^2 (\dot \varphi_1 + \dot \varphi_2) + l_2^2 \dot \varphi_2 - l_1 l_2 (\dot \varphi_1 + 2 \dot \varphi_2) \cos \varphi_1 \bigg) \,,
\eeq
\end{subequations}
and to the friction acting on $C$
\beq
R_C = \frac{1}{2} \mu_C \bigg( l_2^2 \dot \varphi_2^2 + l_3^2 (\dot \varphi_1 + \dot \varphi_2)^2 + 2 l_2 l_3 \dot \varphi_2 (\dot \varphi_1 + \dot \varphi_2) \cos \varphi_1 \bigg) \,,
\eeq
\begin{subequations}
\beq
Q_1^C = -\mu_C \bigg( l_3^2 (\dot \varphi_1 + \dot \varphi_2) + l_2 l_3 \dot \varphi_2 \cos \varphi_1 \bigg)
\eeq
\beq
Q_2^C = -\mu_C \bigg( l_3^2 (\dot \varphi_1 + \dot \varphi_2) + l_2^2 \dot \varphi_2 + l_3 l_2 (\dot \varphi_1 + 2 \dot \varphi_2) \cos \varphi_1 \bigg) \,.
\eeq
\end{subequations}
The terms $A_ {ij}$ in the equations of motion still follow the \eqref{coeff_A}, while $r_{1,2}$ earn some terms due to the presence of frictions
\begin{subequations}
\beq\begin{split}\label{friction_stokes_motion1}
r_1 =& - k_1 \varphi_1 + \Delta l_2 \dot{\varphi}_2^2 \sin \varphi_1 - (\mu_B l_1^2 + \mu_C l_3^2) (\dot \varphi_1 + \dot \varphi_2) + \\
&+ (\mu_B l_1 - \mu_C l_3) l_2 \dot \varphi_2 \cos \varphi_1 \,,
\end{split}\eeq
\beq\begin{split}\label{friction_stokes_motion2}
r_2 =& -\Delta l_2 \dot{\varphi}_1 (\dot{\varphi}_1 + 2 \dot{\varphi}_2) \sin \varphi_1 - F l_2 \sin \varphi_1 - (\mu_B l_1^2 + \mu_C l_3^2) (\dot \varphi_1 + \dot \varphi_2) + \\
&+ (\mu_B l_1 - \mu_C l_3) l_2 (\dot \varphi_1 + 2 \dot \varphi_2) \cos \varphi_1 - (\mu_A + \mu_B + \mu_C) l_2^2 \dot \varphi_2 \,.
\end{split}\eeq
\end{subequations}
We point out that the presence of a general friction does not re-introduce the variable $\varphi_2$ in the equations of motion, so that $\varphi_2$ is still cyclic, but this further non-conservative force destroys the possible Hamiltonian case we had for $F = 0$.

The symmetry of the problem is clearly visible, since the friction coefficients introduce terms analogous to $A_{11}$, $\Delta$ and $M$, with the mass $m_j$ replaced by $\mu_j$. In particular we observe that a term $\mu_B l_1 - \mu_C l_3$ arises in the parameters $r_{1,2}$, formally equal to $\Delta$. We therefore define
\beq\begin{split}
&A_{11\mu} := \mu_B l_1^2 + \mu_C l_3^2 \\
&M_\mu := \mu_A + \mu_B + \mu_C \\
&A_{22\mu} := A_{11\mu} + M_\mu l_2^2 \\
&\Delta_\mu:= \mu_B l_1 - \mu_C l_3
\end{split}\eeq
and require that $\Delta = \Delta_\mu = 0$, so that
\begin{subequations}
\beq
A_{11} = A_{12} = A_{21} = m_B l_1^2 + m_C l_3^2
\eeq
\beq
A_{22} = A_{11} + M l_2^2
\eeq
\beq
r_1 = - k_1 \varphi_1 - A_{11\mu} (\dot \varphi_1 + \dot \varphi_2)
\eeq
\beq
r_2 = - F l_2 \sin \varphi_1 - A_{11\mu} (\dot \varphi_1 + \dot \varphi_2) - M_\mu l_2^2 \dot \varphi_2 \,.
\eeq
\end{subequations}
The equations of motion become
\begin{subequations}\label{syst_friction}
\beq
A_{11} (\ddot \varphi_1 + \ddot \varphi_2) = - k_1 \varphi_1 - A_{11\mu} (\dot \varphi_1 + \dot \varphi_2)
\eeq
\beq\begin{split}
&A_{11} \ddot \varphi_1 + (A_{11} + M l_2^2) \ddot \varphi_2 = - F l_2 \sin \varphi_1 - A_{11\mu} (\dot \varphi_1 + \dot \varphi_2) - M_\mu l_2^2 \dot \varphi_2 \,.
\end{split}\eeq
\end{subequations}
If we subtract the first equation from the second one we get
\beq\label{subtract}
M l_2^2 \ddot \varphi_2 = k_1 \varphi_1 - F l_2 \sin \varphi_1 - M_\mu l_2^2 \dot \varphi_2 \,.
\eeq
If the following condition on the coefficients
\beq\label{sym_mmu}
\frac{A_{11}}{M} = \frac{A_{11\mu}}{M_\mu}
\eeq
holds, a $\varphi_2$-independent equation can be extracted from the system, by substituting the \eqref{subtract} in the first equation of \eqref{syst_friction}
\beq\label{syst_friction_sym1}
\ddot \varphi_1 = - \frac{k_1}{M l_2^2} \bigg( 1 + \frac{M l_2^2}{A_{11}} \bigg) \varphi_1 - \frac{F}{M l_2} \sin \varphi_1 - \frac{A_{11\mu}}{A_{11}} \dot \varphi_1 \,,
\eeq
that is the equation of a one-dimensional damped harmonic oscillator. We then substitute the previous one in the second equation of \eqref{syst_friction} to get
\beq\label{syst_friction_sym2}
\ddot \varphi_2 = \frac{k_1}{M l_2^2} \varphi_1 - \frac{F}{M l_2} \sin \varphi_1 - \frac{A_{22\mu}}{A_{22}} \dot \varphi_2
\eeq
and the system given by equations \eqref{syst_friction_sym1} and \eqref{syst_friction_sym2} is formally integrable.

In Figure \ref{friction_stokes_spiral} is shown the expected spiral motion, that is the system has an attractive point. In the case of breaking of the three symmetries considered above, i.e. $\Delta = \Delta_\mu = 0$ and \eqref{sym_mmu}, also limit cycles can arise, as shown in Figure \ref{friction_stokes_periodic1} and \ref{friction_stokes_periodic2}; furthermore, the approaching to the attractive point can be unusual, as shown in Figure \ref{friction_stokes_strange}.
\begin{figure}
\centering
\includegraphics[width=7cm]{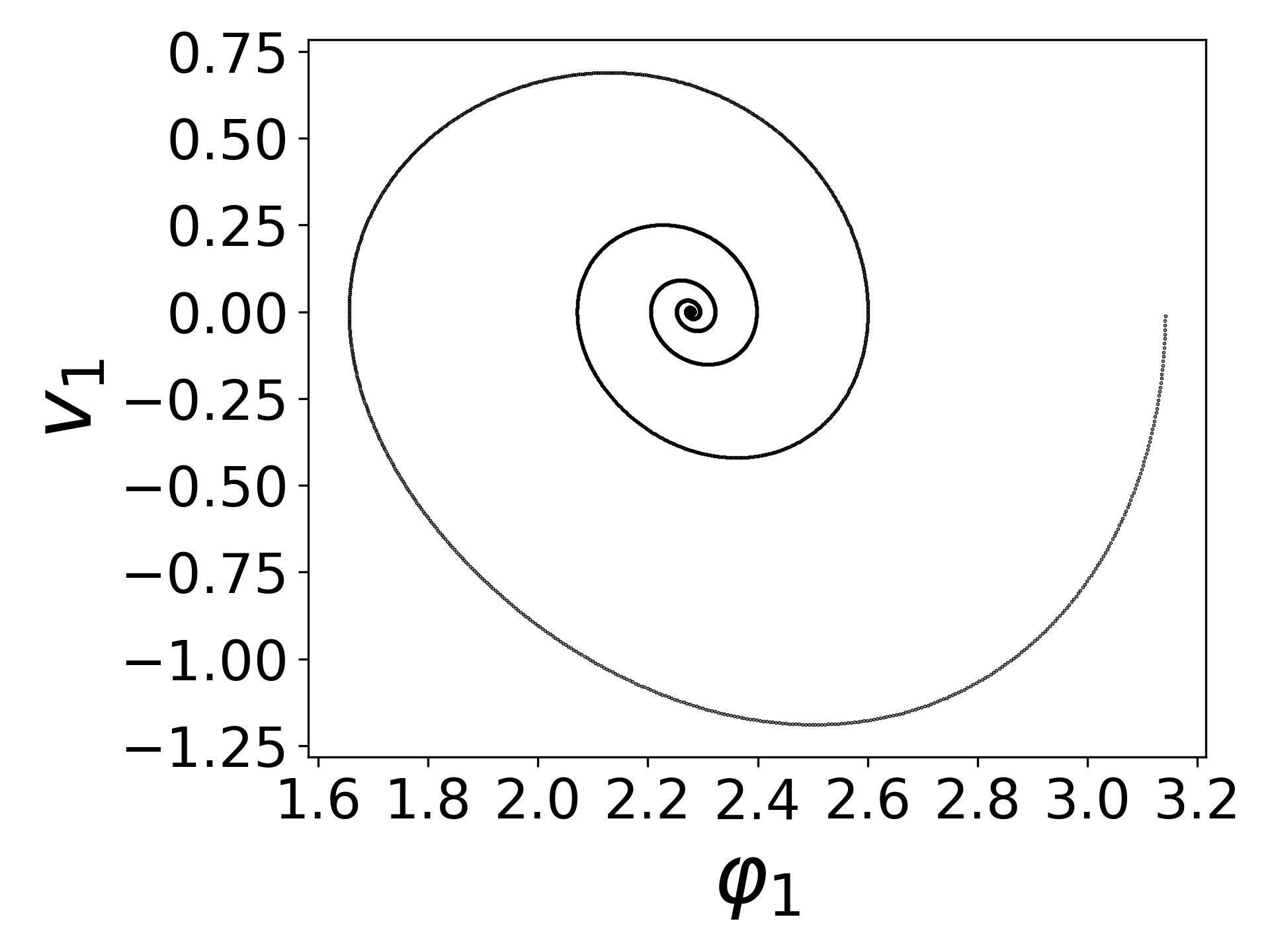}
\caption{Attractive point in presence of Stokes friction, $\Delta = 0$, $\Delta_\mu = 0$ and $\frac{A_{11}}{M} = \frac{A_{11\mu}}{M_\mu}$. Parameters: $m_A = 1$, $m_B = 2$, $m_C = 1$,  $l_1 = 1$, $l_2 = 1$, $l_3 = 2$, $k_1 = 2$, $F = 10$, $\mu_A = 0.5$, $\mu_B = 1$, $\mu_C = 0.5$. Initial conditions: $\varphi_1(0) = \pi$, $\varphi_2(0) = 0$, $v_1(0) = v_2(0) = 0$.}
\label{friction_stokes_spiral}
\end{figure}
\begin{figure}
\centering
\includegraphics[width=7cm]{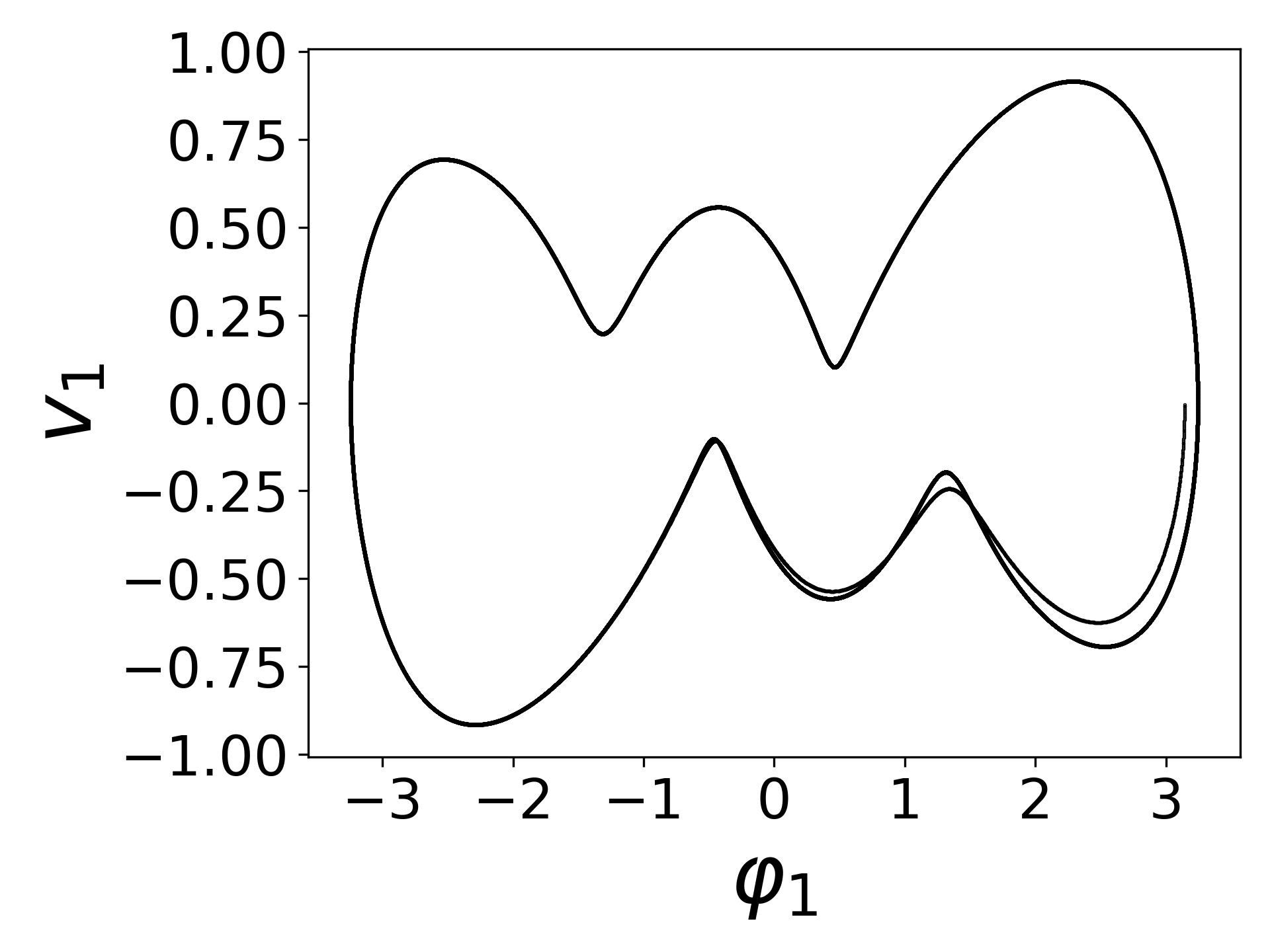}
\caption{Limit cycle in presence of Stokes friction, $\Delta \ne 0$, $\Delta_\mu \ne 0$ and $\frac{A_{11}}{M} \ne \frac{A_{11\mu}}{M_\mu}$. Parameters: $m_A = 1$, $m_B = 2.5$, $m_C = 5$,  $l_1 = 1$, $l_2 = 1$, $l_3 = 2$, $k_1 = 2$, $F = 10$, $\mu_A = 0.4$, $\mu_B = 2$, $\mu_C = 0$. Initial conditions: $\varphi_1(0) = \pi$, $\varphi_2(0) = 0$, $v_1(0) = v_2(0) = 0$.}
\label{friction_stokes_periodic1}
\end{figure}
\begin{figure}
\centering
\includegraphics[width=7cm]{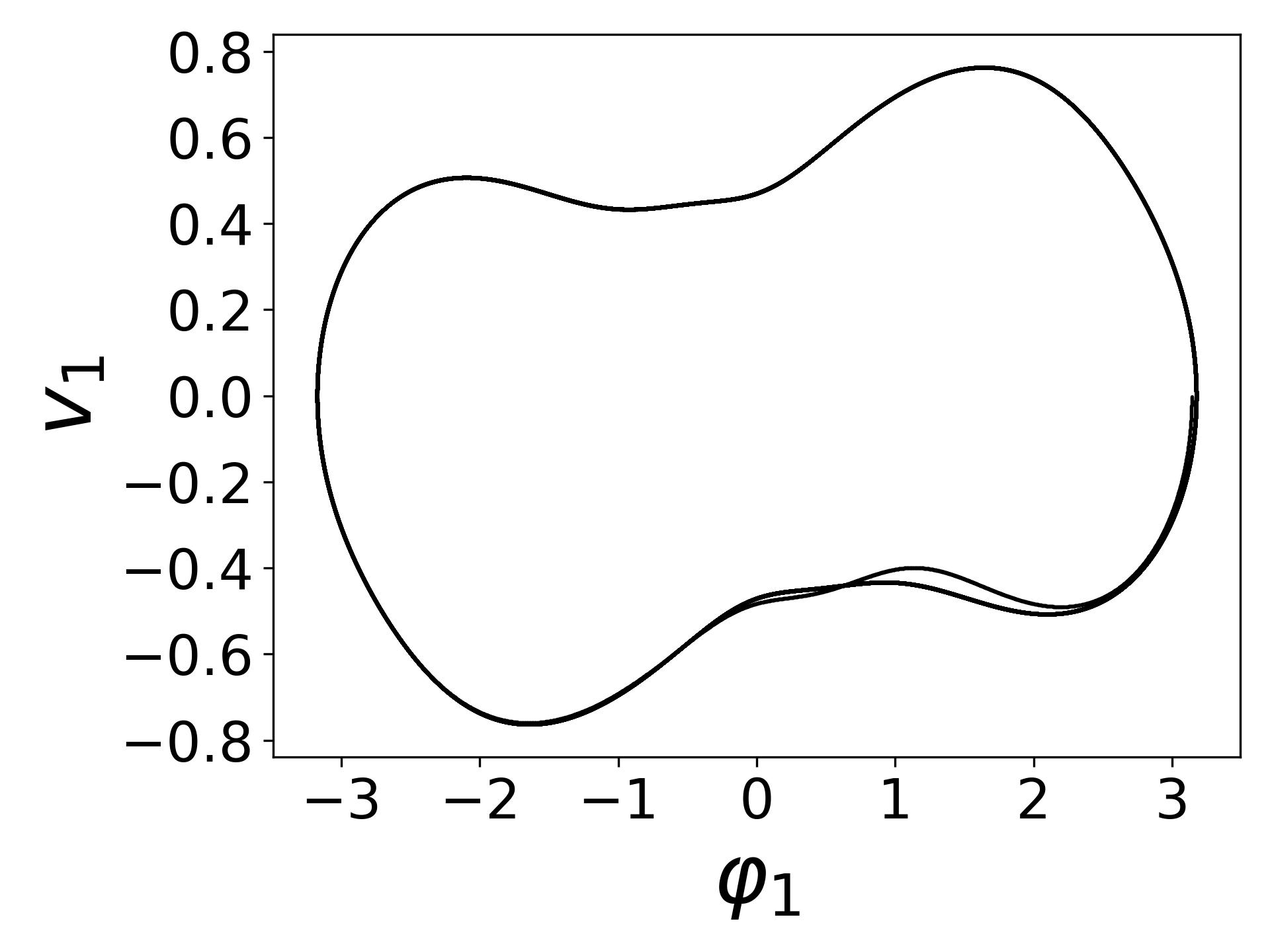}
\caption{Limit cycle in presence of Stokes friction, $\Delta \ne 0$, $\Delta_\mu \ne 0$ and $\frac{A_{11}}{M} \ne \frac{A_{11\mu}}{M_\mu}$. Parameters: $m_A = 1$, $m_B = 2.5$, $m_C = 5$,  $l_1 = 1$, $l_2 = 2$, $l_3 = 2$, $k_1 = 2$, $F = 10$, $\mu_A = 0.4$, $\mu_B = 2$, $\mu_C = 0.5$. Initial conditions: $\varphi_1(0) = \pi$, $\varphi_2(0) = 0$, $v_1(0) = v_2(0) = 0$.}
\label{friction_stokes_periodic2}
\end{figure}
\begin{figure}
\centering
\subfloat[][\label{friction_stokes_strange}]
{\includegraphics[width=.45\textwidth]{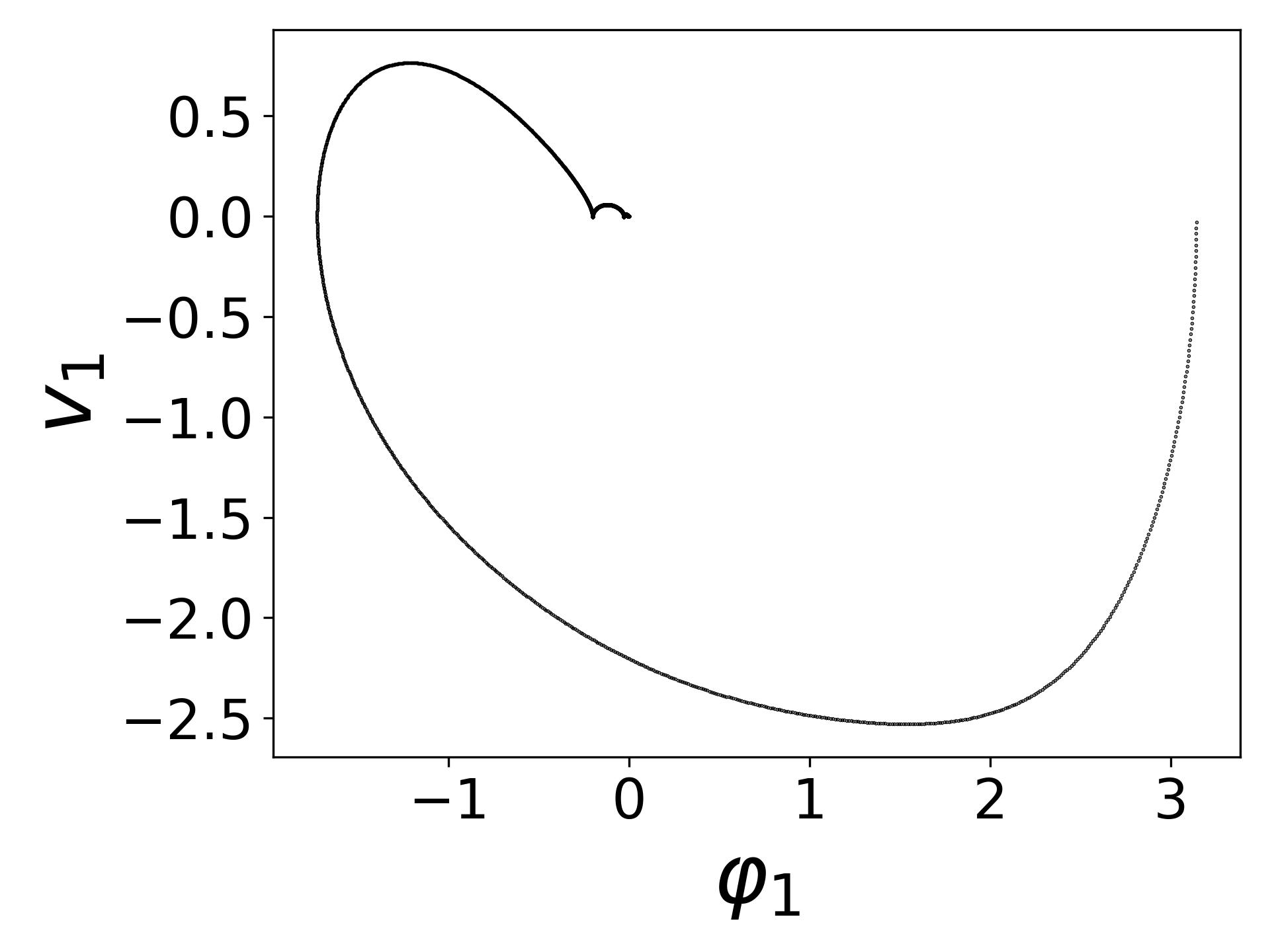}} \\
\subfloat[][\label{friction_stokes_strange_zoom1}]
{\includegraphics[width=.45\textwidth]{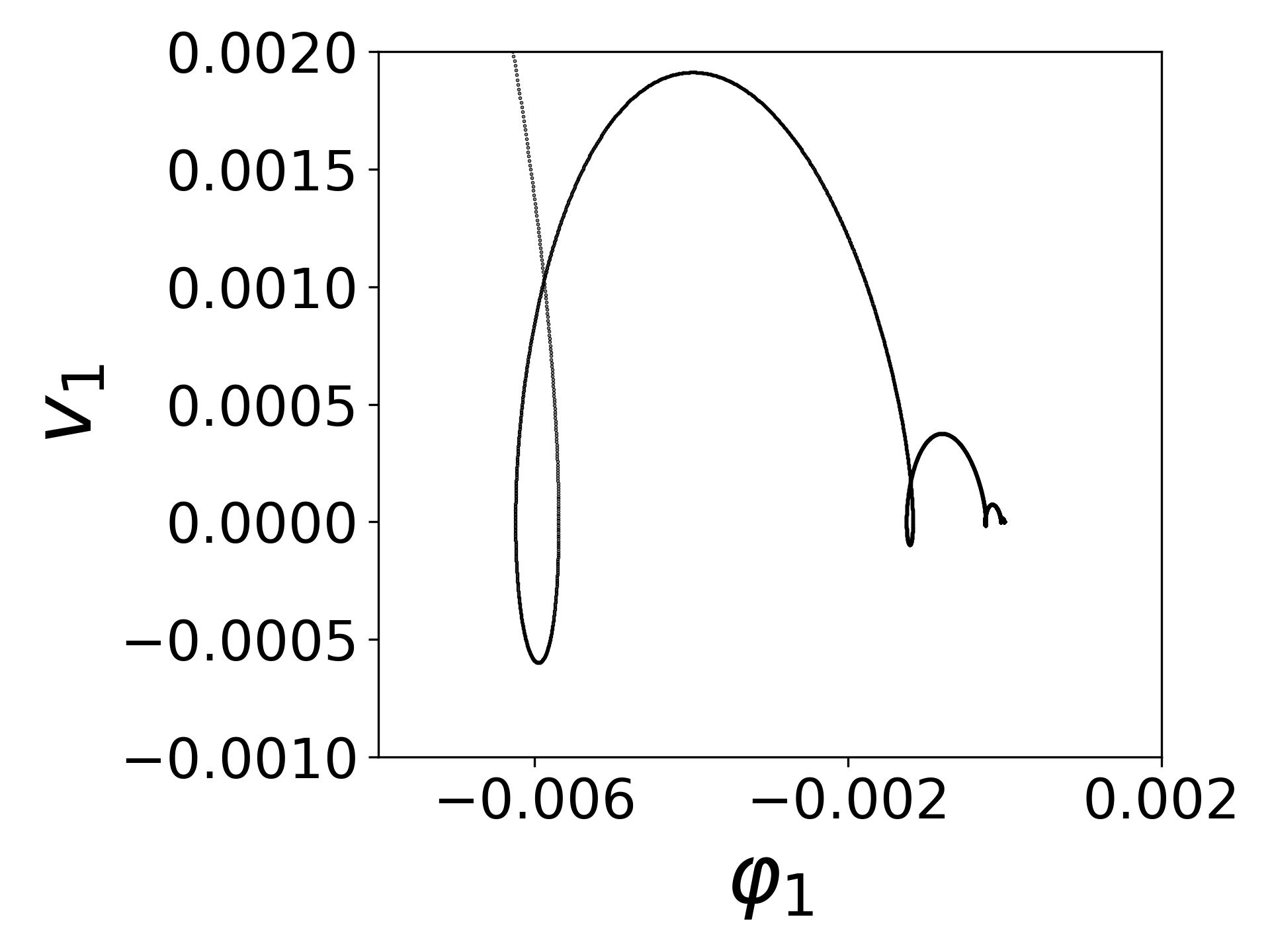}} \quad
\subfloat[][\label{friction_stokes_strange_zoom2}]
{\includegraphics[width=.45\textwidth]{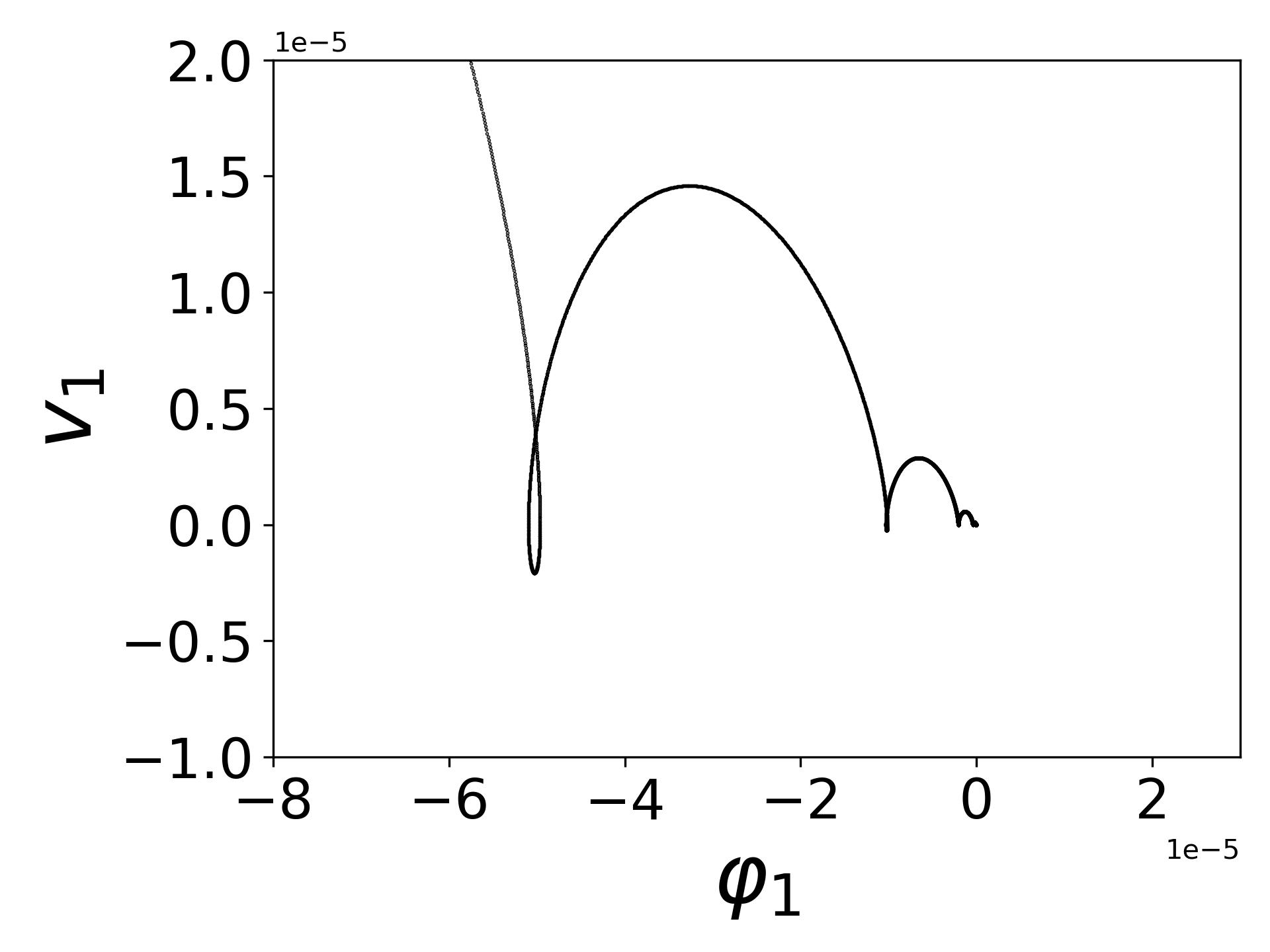}}
\caption{"Jumping" motion approaching to an attractive point in presence of Stokes friction; two zooms up to a scale of $10^{-5}$ are shown. Parameters: $m_A = 1$, $m_B = 1.5$, $m_C = 1$,  $l_1 = 1.1$, $l_2 = 3$, $l_3 = 1$, $k_1 = 4$, $F = 10$, $\mu_A = 0.5$, $\mu_B = 0.5$, $\mu_C = 0.5$. Initial conditions: $\varphi_1(0) = \pi$, $\varphi_2(0) = 0$, $v_1(0) = v_2(0) = 0$.}
\label{friction_stokes_strange}
\end{figure}

\subsection{Friction on the pins}
Let us now consider two friction forces generated on the pins $O$ and $A$ represented by
\beq
\textbf{A} = - \mu \dot \varphi \hat \varphi, \qquad \mu \ge 0
\eeq
that resuly in adding a torque opposite to the rotation of the system. The friction on the pin $O$ acts on the rod $\overline{OA}$, so it acts on the point $A$; the friction on the pin $A$ acts on the rod $\overline{BC}$, so it acts on the center of mass of the material points $B$ and $C$. In this case it is more useful to calculate the generalized forces by using the definition. \\We calculate the generalized force associated to the friction acting on the pin $O$
\beq
\textbf{A}_O = - \mu_O \dot \varphi_2 \hat \varphi_2 = - \mu_O \dot \varphi_2
\begin{pmatrix}
- \sin \varphi_2 \\ \cos \varphi_2
\end{pmatrix} \,,
\eeq
\begin{subequations}
\beq
Q_1^O = \displaystyle \textbf{A}_O \cdot \frac{\partial \textbf{r}_A}{\partial \varphi_1} = 0
\eeq
\beq
Q_2^O = \displaystyle \textbf{A}_O \cdot \frac{\partial \textbf{r}_A}{\partial \varphi_2} = \mu_O l_2 \dot \varphi_2 \cos (2 \varphi_2)
\eeq
\end{subequations}
and the generalized force associated to the friction acting on the pin $A$
\beq
\textbf{A}_A = - \mu_A \dot \varphi_1 \hat \varphi_1 = - \mu_A \dot \varphi_1
\begin{pmatrix}
- \sin (\varphi_1 + \varphi_2) \\ \cos (\varphi_1 + \varphi_2)
\end{pmatrix} \,,
\eeq
\beq
\textbf{r}_{BC} = \frac{m_B \textbf{r}_B + m_C \textbf{r}_C}{m_B + m_C} =
\begin{pmatrix}
l_2 \cos \varphi_2 - \frac{\Delta}{m_B + m_C} \cos(\varphi_1 + \varphi_2) \\ - l_2 \sin \varphi_2 + \frac{\Delta}{m_B + m_C} \sin(\varphi_1 + \varphi_2)
\end{pmatrix} \,,
\eeq
\begin{subequations}
\beq
Q_1^A = \displaystyle \textbf{A}_A \cdot \frac{\partial \textbf{r}_{BC}}{\partial \varphi_1} = - \frac{\mu_A \Delta}{m_B + m_C} \dot \varphi_1 \cos( 2 \varphi_1 + 2 \varphi_2 ) \\
\eeq
\beq
Q_2^A = \displaystyle \textbf{A}_A \cdot \frac{\partial \textbf{r}_{BC}}{\partial \varphi_2} = - \mu_A \dot \varphi_1 \bigg( -l_2 \cos (\varphi_1 + 2 \varphi_2) + \frac{\Delta}{m_B + m_C} \cos( 2 \varphi_1 + 2 \varphi_2) \bigg) \,.
\eeq
\end{subequations}
The equations of motion change as follows: the terms $A_{ij}$ still follow the \eqref{coeff_A}, while $r_{1,2}$ earn some terms due to the presence of frictions
\begin{subequations}
\beq
r_1 = - k_1 \varphi_1 + \Delta l_2 \dot \varphi_2^2 \sin \varphi_1 - \frac{\mu_A \Delta}{m_B + m_C} \dot \varphi_1 \cos( 2 \varphi_1 + 2 \varphi_2 )
\eeq
\beq\begin{split}
r_2 =& - \Delta l_2 \dot \varphi_1 (\dot \varphi_1 + 2 \dot \varphi_2) \sin \varphi_1 - F l_2 \sin \varphi_1 + \mu_O l_2 \dot \varphi_2 \cos (2 \varphi_2) + \\
&+ \mu_A l_2 \dot \varphi_1 \cos (\varphi_1 + 2 \varphi_2) - \frac{\mu_A \Delta}{m_B + m_C} \dot \varphi_1 \cos( 2 \varphi_1 + 2 \varphi_2) \,.
\end{split}\eeq
\end{subequations}
In this case there's no way to obtain the previous integrable cases, so we may expect that the system shows a chaotic behavior, with the holding or the breaking of the known symmetries.

Anyway it could be interesting to analyze the behavior of the system under a progressive breaking of the symmetry $\Delta = 0$; we will do that by holding fixed the initial conditions and all the parameters except for $m_B$, that will be progressively increased. As shown in Figure \ref{friction_pins_delta0}, with the holding of the symmetry the motion starts from the initial point and seems trying to build a limit cycle that expands more and more, up to generate irregular motion. A first small breaking of the symmetry, shown in Figure \ref{friction_pins_delta01}, gives rise to a very different behavior; in particular the chaotic motion accumulates on pseudo-spherical shapes, symmetric with respect to the axis $\varphi_1 = 0$. This behavior is mixed with the previous one under a further increasing of $\Delta$, as shown in Figure \ref{friction_pins_delta02}. Notice that this pseudo-spherical shape has been already found in \cite{Polekhin}, for a general non-integrable case of the system. By increasing $\Delta$, as shown in Figure \ref{friction_pins_delta05}, the system seems again trying to build limit cycles that expand as the time increases, but in a more complex way, while the oscillations of the motion are now more visible. A further increase of $\Delta$ gives rise again to a mixed behavior, shown in Figure \ref{friction_pins_delta1}; we have now a pseudo-cyclic limit interrupted by the presence of two new chaotic orbits, similar to the previous ones.
\begin{figure}[H]
\centering
\subfloat[][\emph{$\Delta = m_B - 1 = 0$.} \label{friction_pins_delta0}]
{\includegraphics[width=.45\textwidth]{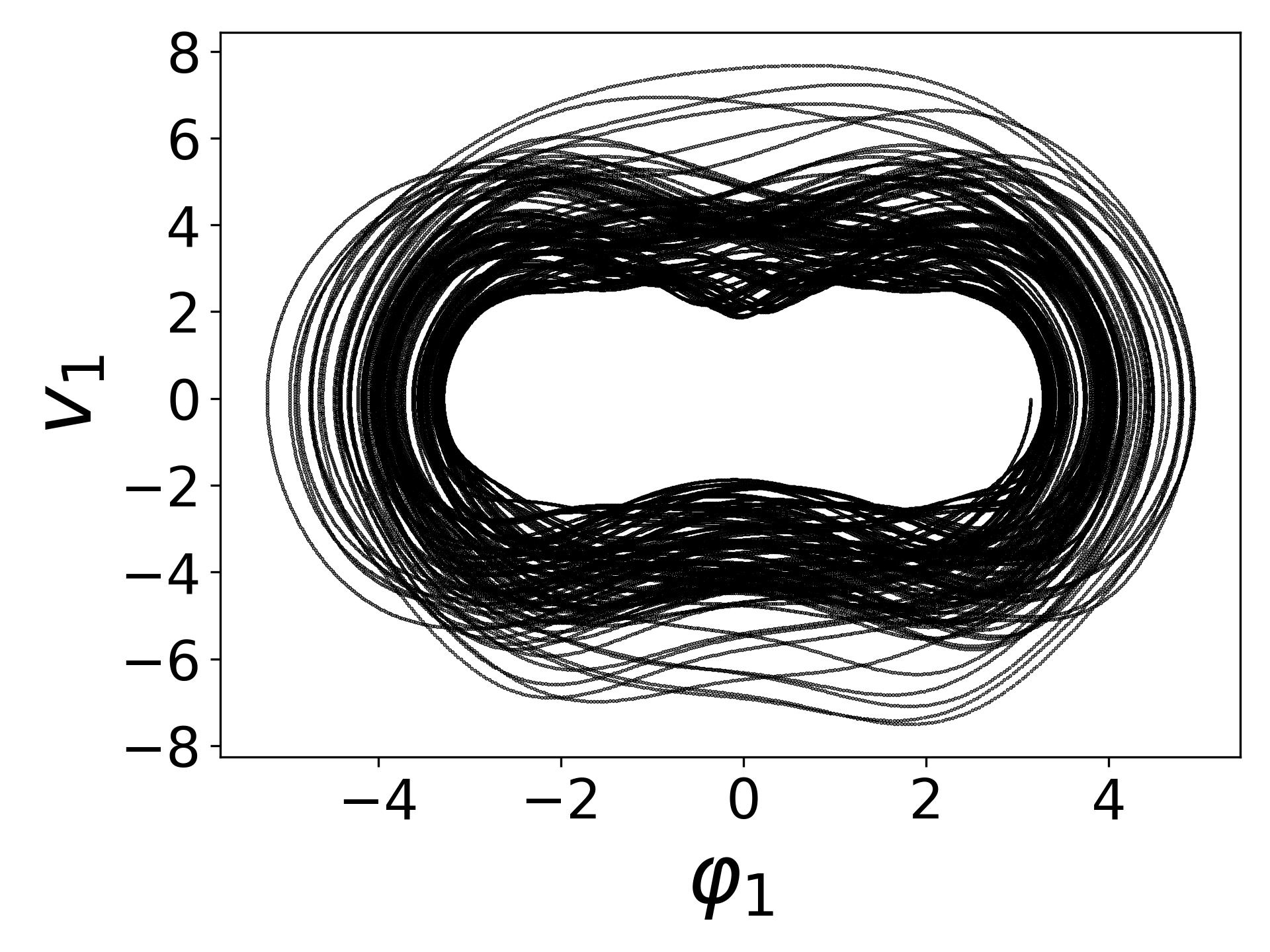}} \quad
\subfloat[][\emph{$\Delta = m_B - 1 = 0.1$.} \label{friction_pins_delta01}]
{\includegraphics[width=.45\textwidth]{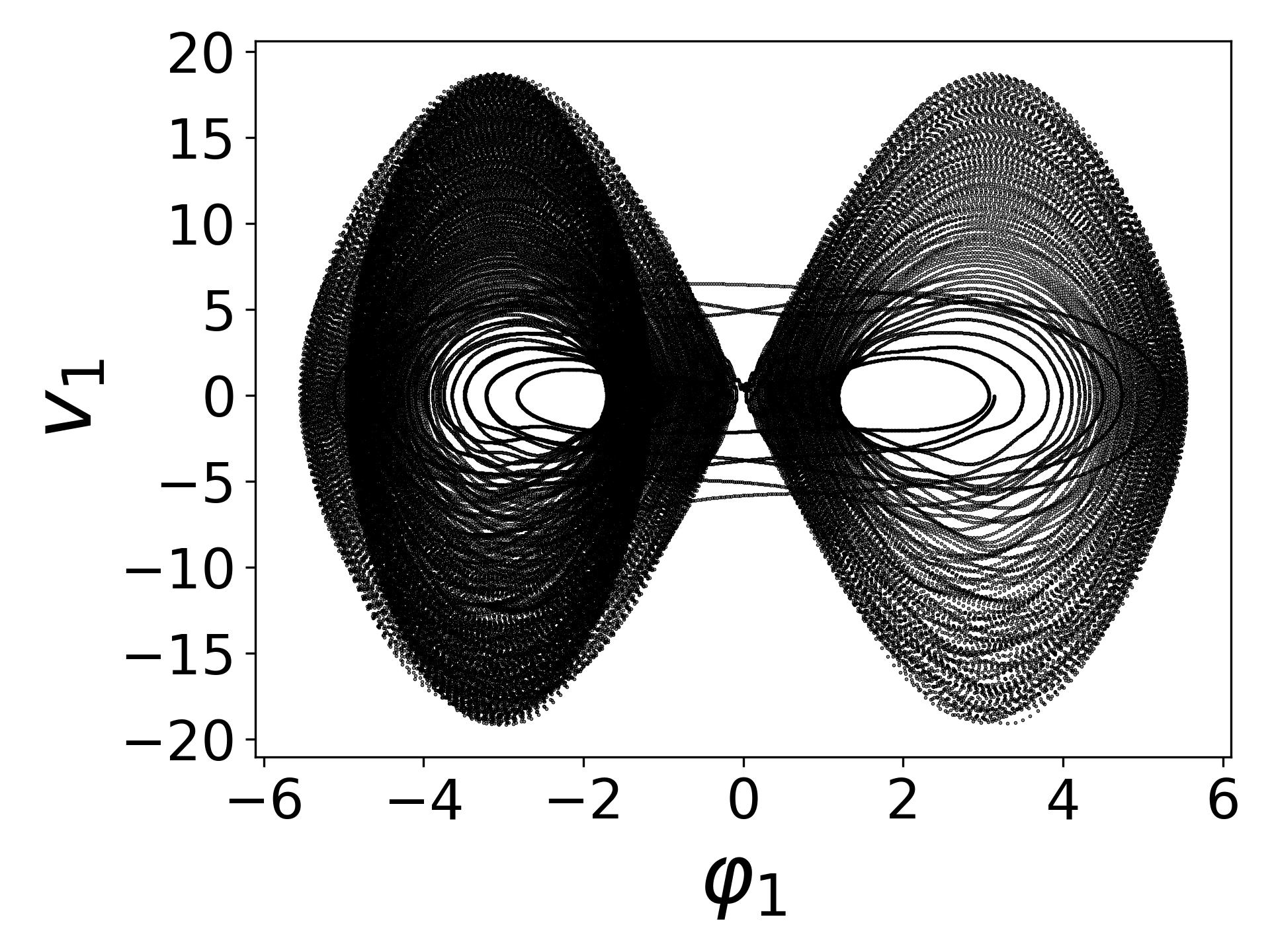}} \\
\subfloat[][\emph{$\Delta = m_B - 1 = 0.2$.} \label{friction_pins_delta02}]
{\includegraphics[width=.45\textwidth]{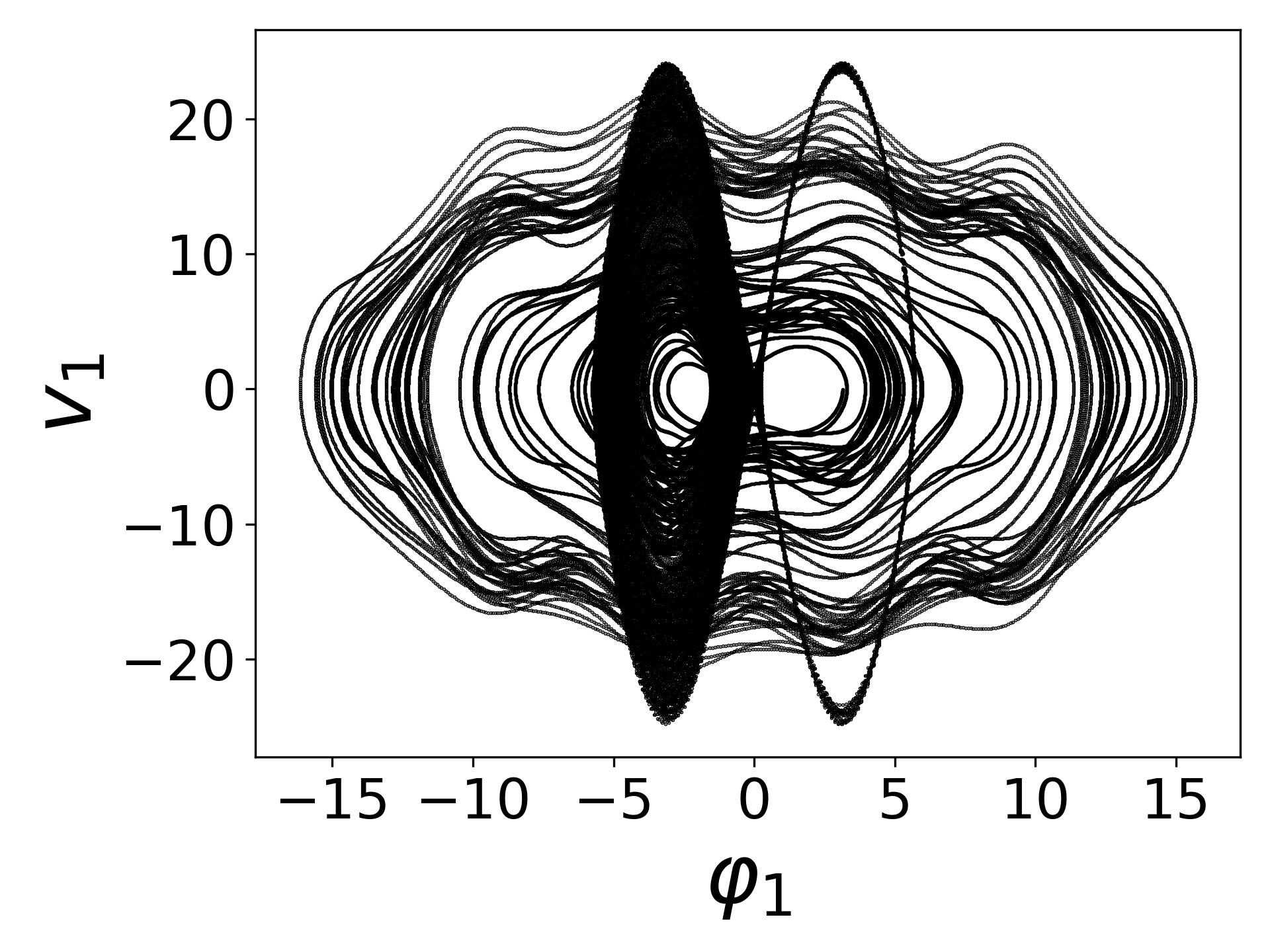}} \quad
\subfloat[][\emph{$\Delta = m_B - 1 = 0.5$.} \label{friction_pins_delta05}]
{\includegraphics[width=.45\textwidth]{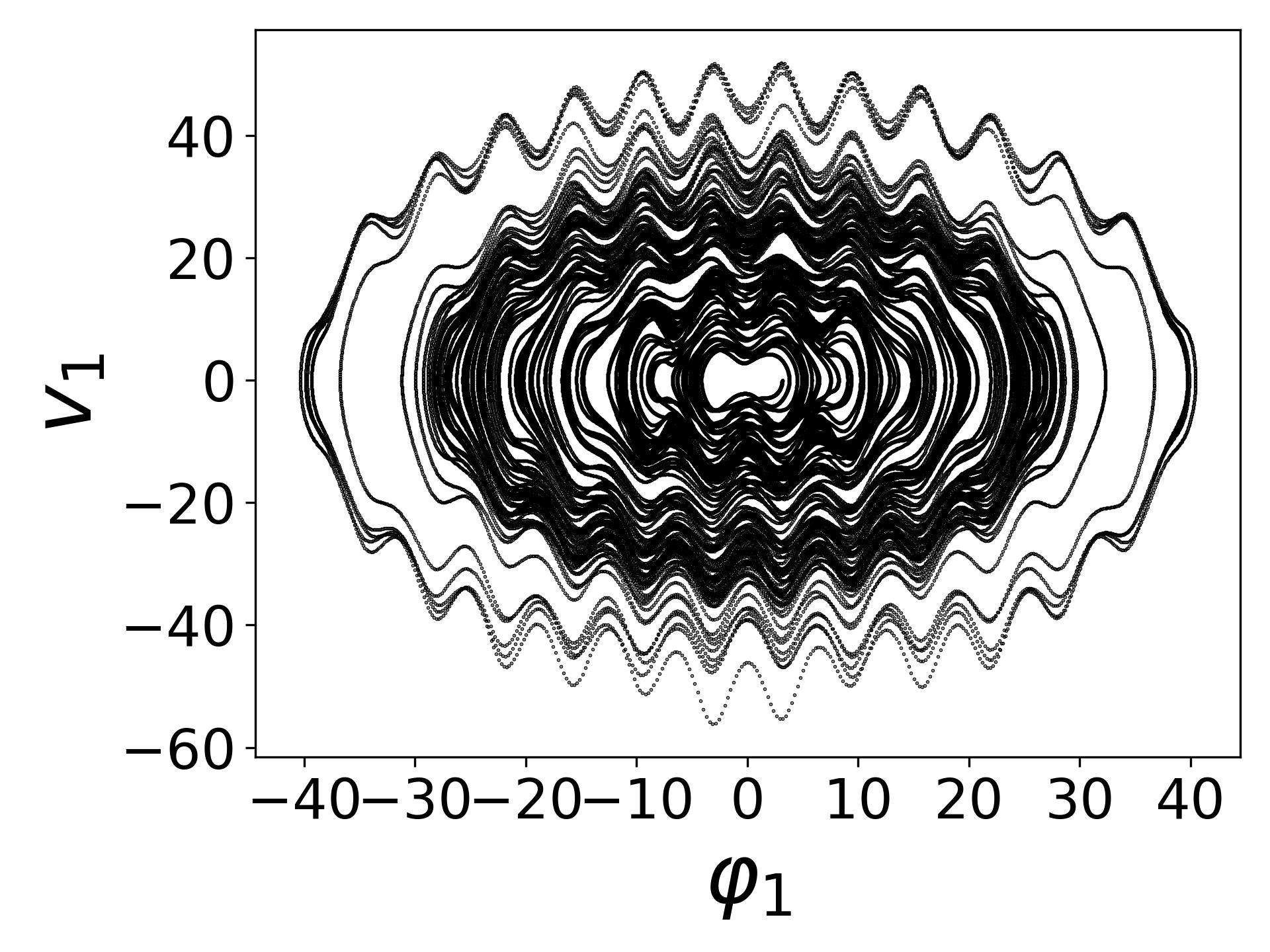}} \\
\subfloat[][\emph{$\Delta = m_B - 1 = 1$.} \label{friction_pins_delta1}]
{\includegraphics[width=.45\textwidth]{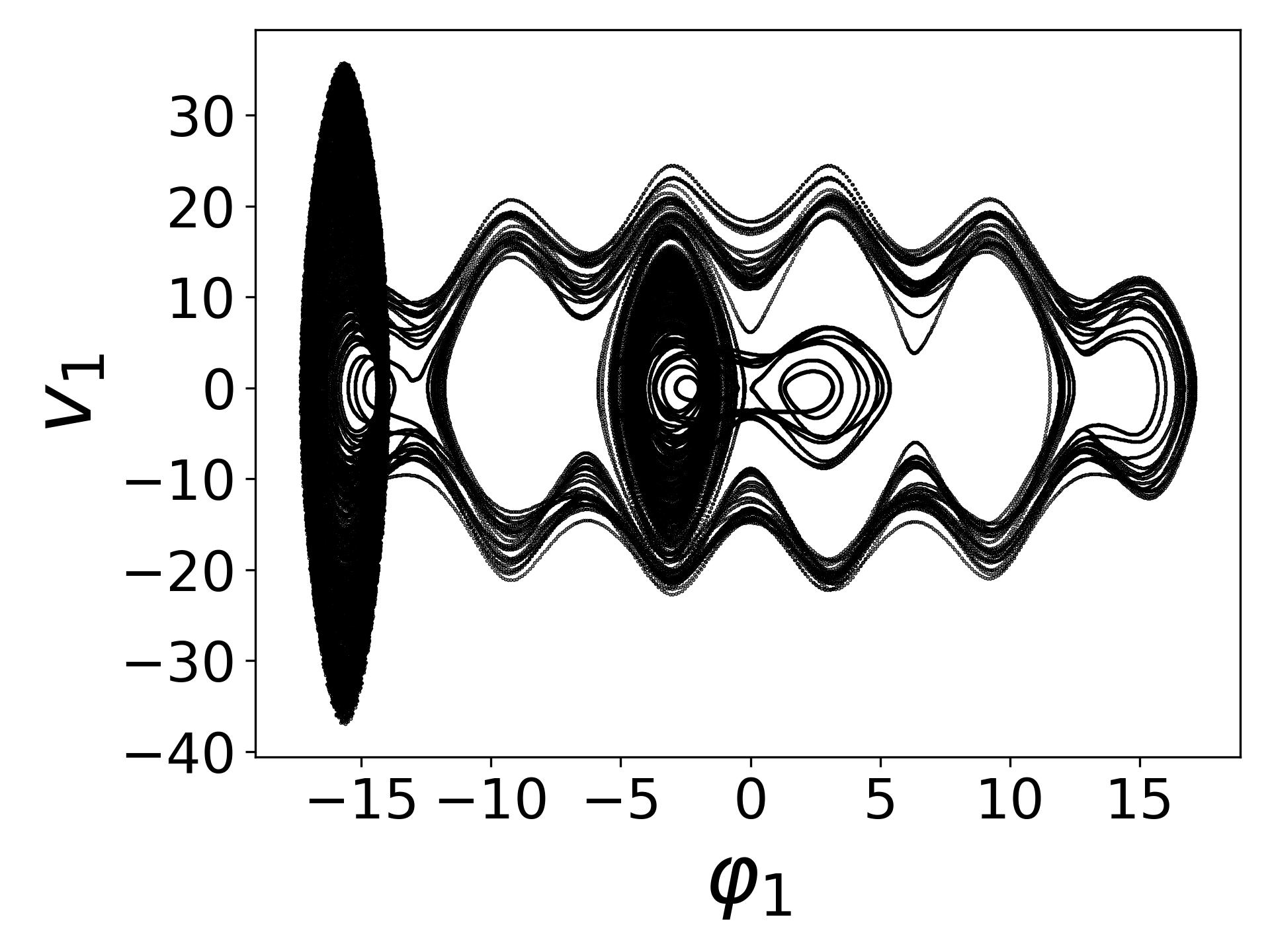}}
\caption{Progressive breaking of the symmetry $\Delta = 0$. Parameters: $m_A = 1$, $m_C = 1$,  $l_1 = 1$, $l_2 = 1$, $l_3 = 1$, $k_1 = 2$, $F = 10$, $\mu_O = 1$, $\mu_A = 2$. Initial conditions: $\varphi_1(0) = \pi$, $\varphi_2(0) = 0$, $v_1(0) = v_2(0) = 0$.}
\label{friction_pins_prog_delta}
\end{figure}
The simulations shown in Figure \ref{friction_pins_prog_delta} suggest that, for relatively small breaking of the symmetry $\Delta = 0$, the system mixes two different behaviors: a phase in which it tries to build limit cycles that expand as the time increase, and a phase in which the completely chaotic motion is dominant and it accumulates on a well defined pseudo-spherical shape. This behavior is confirmed even if we increase $\Delta$ (i.e. $m_B$ for our considerations) up to relatively large values; the system keeps to mix the two previous behaviors, with the chaotic orbit that moves along the $\varphi_1$-axis and the possible rising of a second smaller chaotic orbit, even if this progressively leads to a deformation of the pseudo-spherical shape.

A further suggested interpretation is the following: from Figure \ref{friction_pins_delta01} and \ref{friction_pins_delta02}, we see that the motion accumulates on a symmetrical shape, by going from one side to the other (for a sufficiently large number of iterations we would see that in Figure \ref{friction_pins_delta02} the right side will be filled). This back and forth motion reminds the Lorenz strange attractor \cite{Lorenz}, and the fact that this finite portion of the phase space is recurrent for several choices of the parameters could suggest the presence of a strange attractor for the Ziegler pendulum; the appearence of this attractor is actually different from the usual one, since the motion is not totally accumulated on it, as is (for instance) for the Lorenz attractor.

Let us now briefly look for an estimate of the threshold values of the friction coefficients for the transition to chaos. If we set $F = 0$ and $\mu_O = \mu_A = 0$, we of course obtain the integrable Hamiltonian case studied in \cite{Polekhin} and the related periodic orbits. By keeping $F = 0$, as shown in Figure \ref{threshold_mu01}, a small value for $\mu_O$ and $\mu_A$ is sufficient to transform the periodic orbit in a quasi-periodic motion, then the system passes through several phases, more or less chaotic; in particular, as $\mu_O$ and $\mu_A$ increase, the initial closed orbit  is split in several and more complex semi-orbits, up to the fast rising of an effective chaotic motion. The value $\mu_O = \mu_A = 0.59$ can be seen as a threshold value for the system, in the sense that it distinguishes between an approximately regular motion (Figure \ref{threshold_mu01} - \ref{threshold_mu05}) and an effective chaotic motion (Figure \ref{threshold_mu0591} - \ref{threshold_mu06}). Furthermore, a comparison between the Lyapunov exponents associated to the different choices of parameters suggests to interpret this value as related to a maximal chaotic orbit; indeed a Lyapunov exponent $\lambda_{\varphi_1} \sim 0.25$ is associated to the case $\mu_O = \mu_A = 0.59$, while a smaller value is related to lesser and greater friction coefficients.
\begin{figure}[H]
\centering
\subfloat[][\emph{$\mu_O = \mu_A = 0.1$.} \label{threshold_mu01}]
{\includegraphics[width=.45\textwidth]{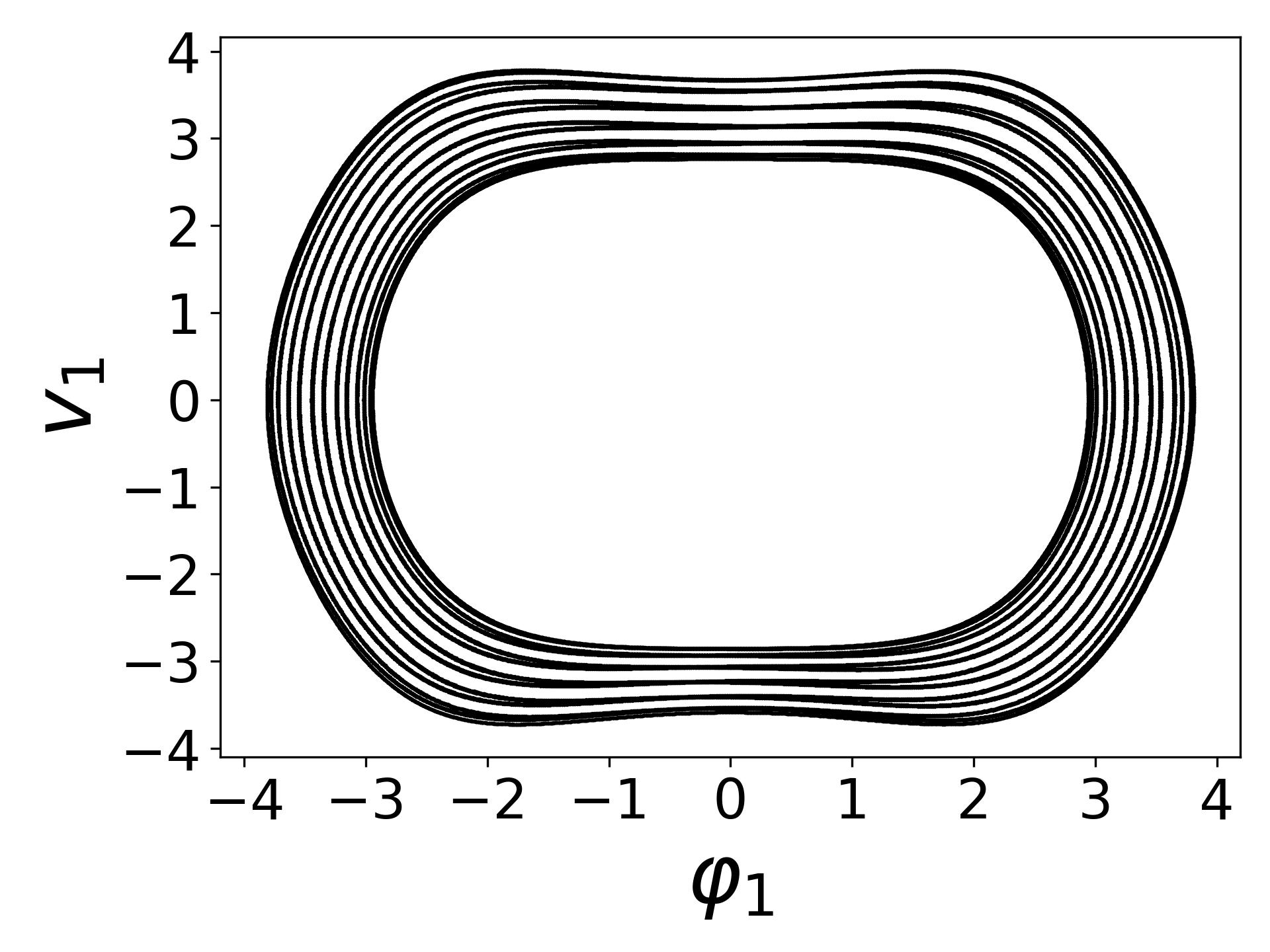}} \quad
\subfloat[][\emph{$\mu_O = \mu_A = 0.5$.} \label{threshold_mu05}]
{\includegraphics[width=.45\textwidth]{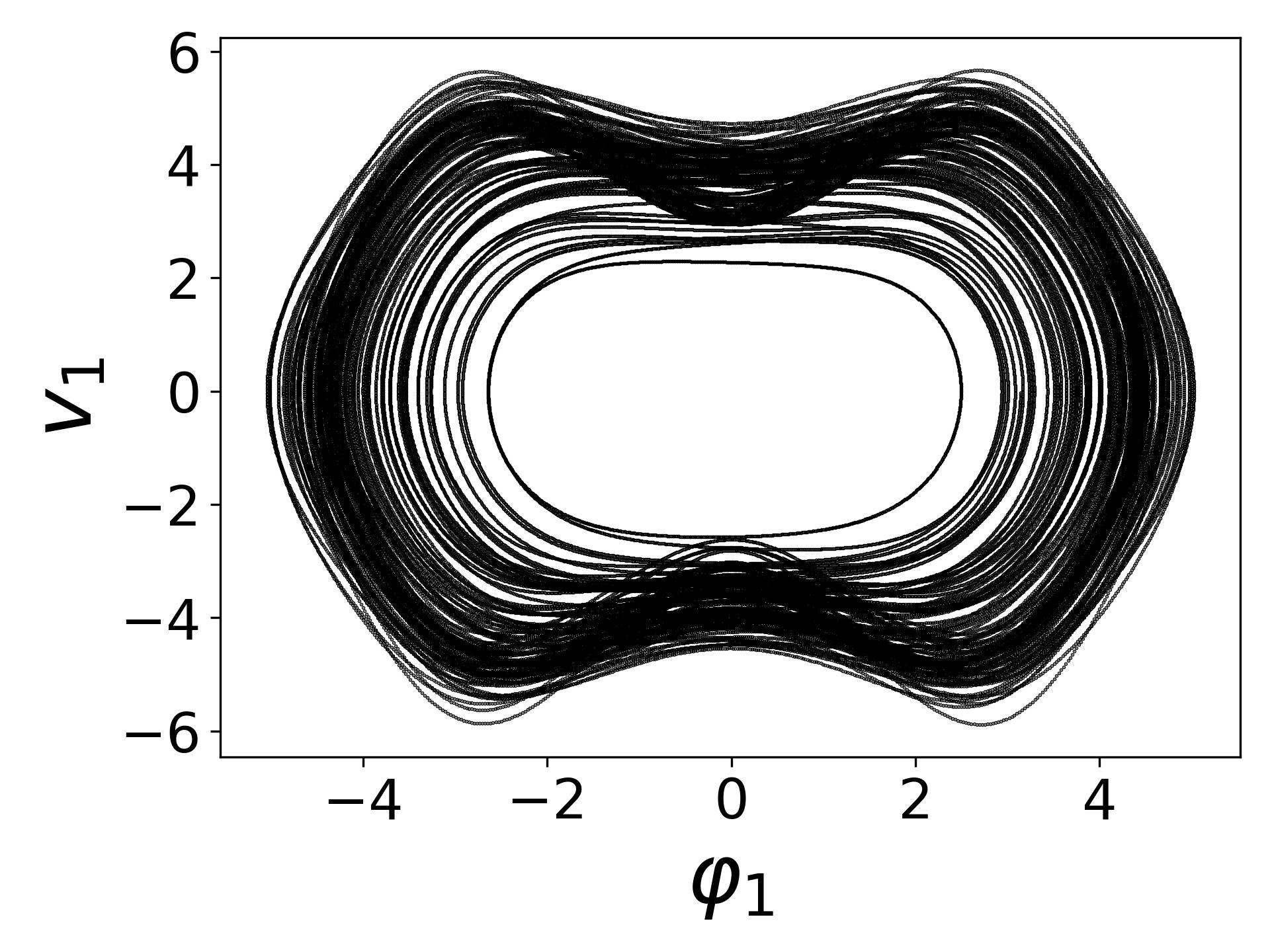}} \\
\subfloat[][\emph{$\mu_O = \mu_A = 0.59$.} \label{threshold_mu059}]
{\includegraphics[width=.45\textwidth]{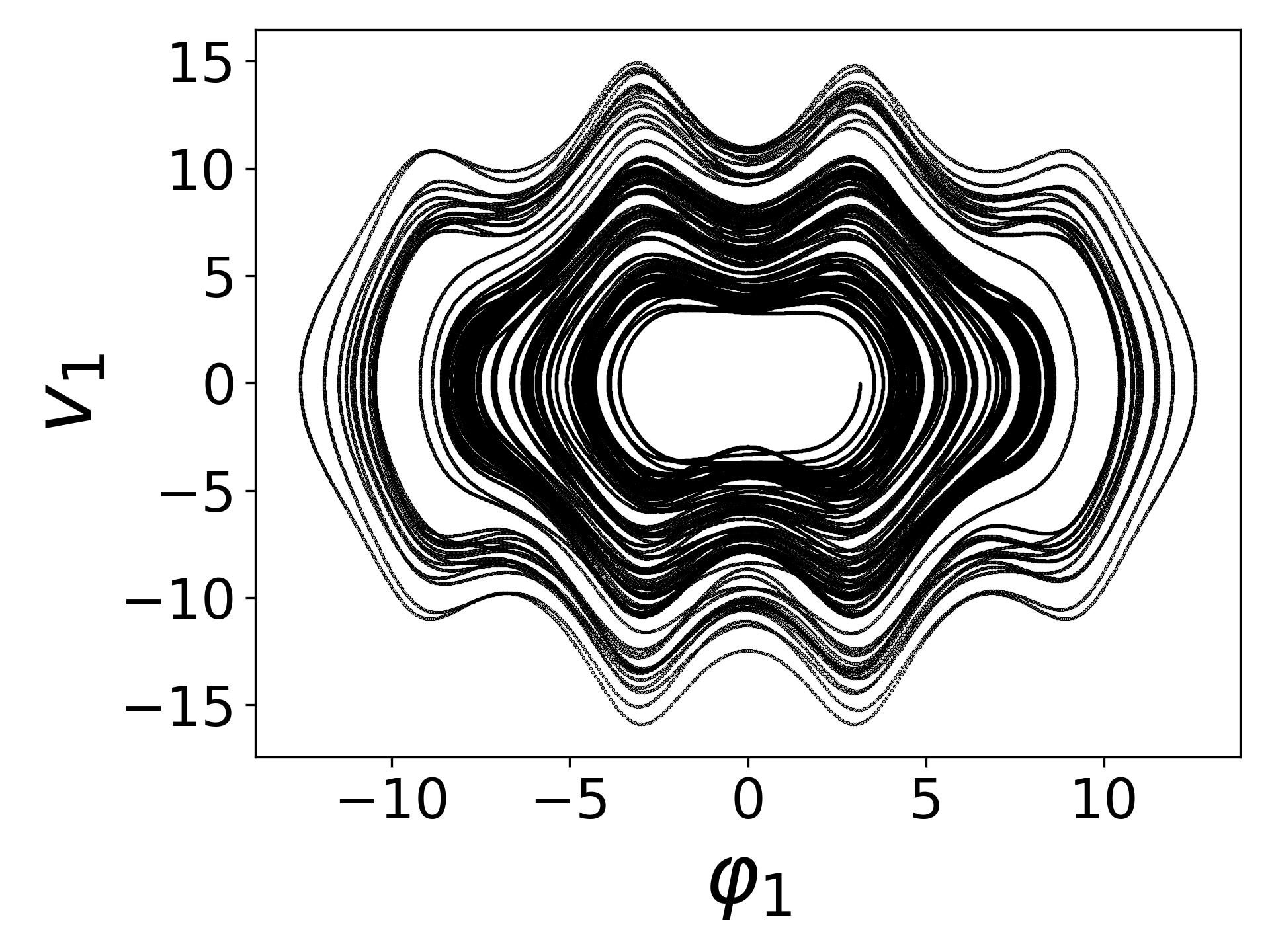}} \quad
\subfloat[][\emph{$\mu_O = \mu_A = 0.591$.} \label{threshold_mu0591}]
{\includegraphics[width=.45\textwidth]{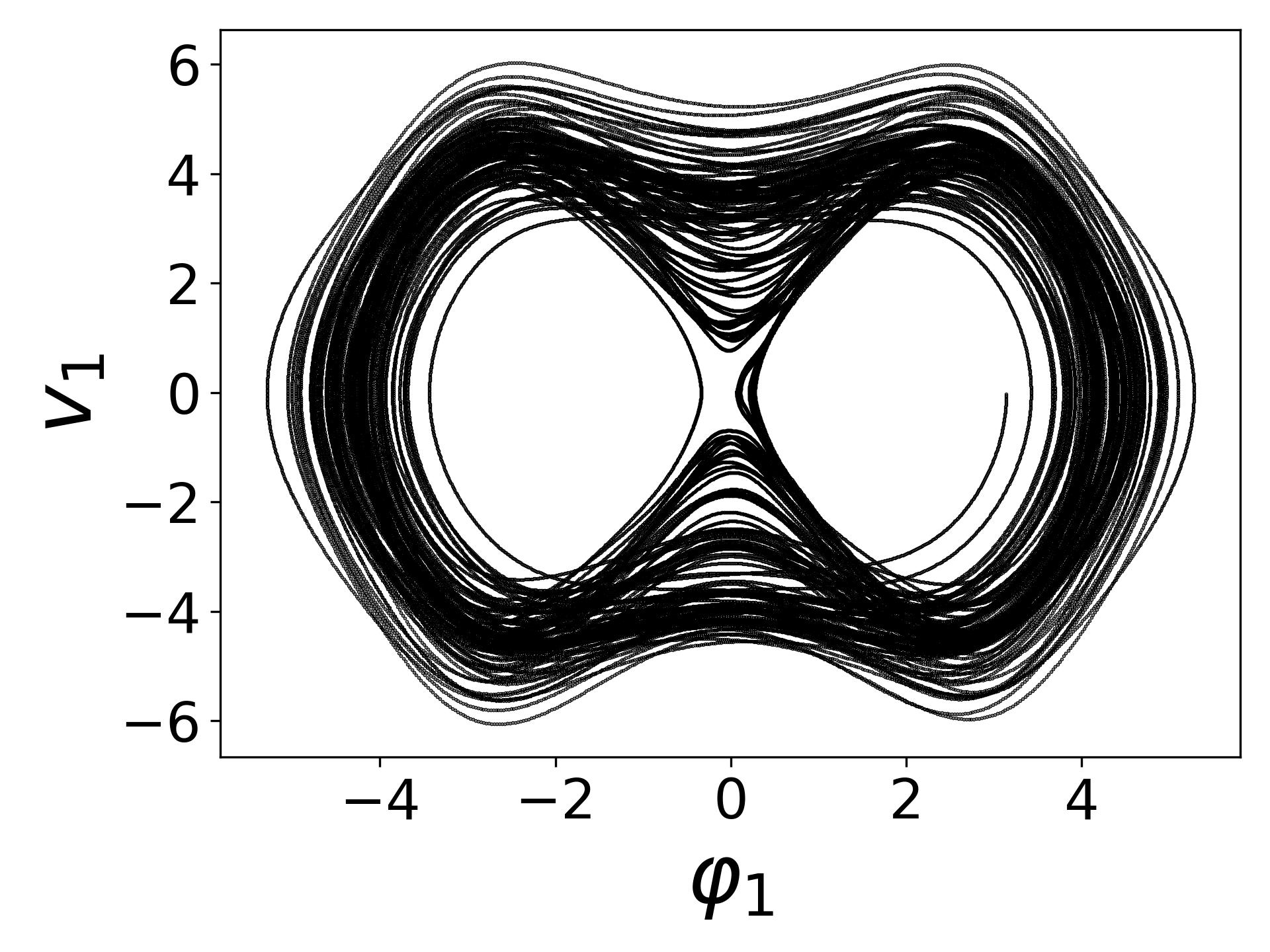}} \\
\subfloat[][\emph{$\mu_O = \mu_A = 0.595$.} \label{threshold_mu0595}]
{\includegraphics[width=.45\textwidth]{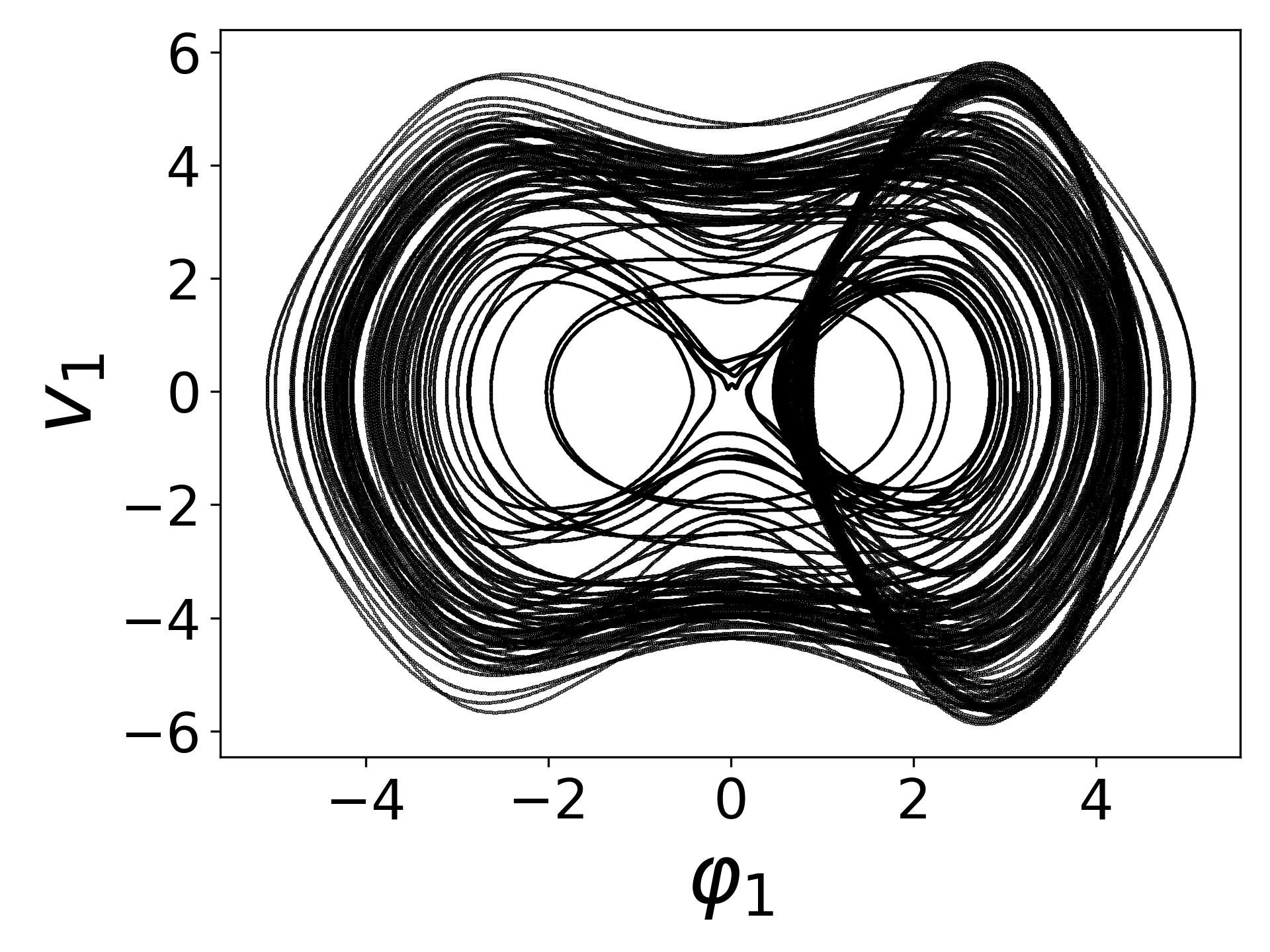}} \quad
\subfloat[][\emph{$\mu_O = \mu_A = 0.6$.} \label{threshold_mu06}]
{\includegraphics[width=.45\textwidth]{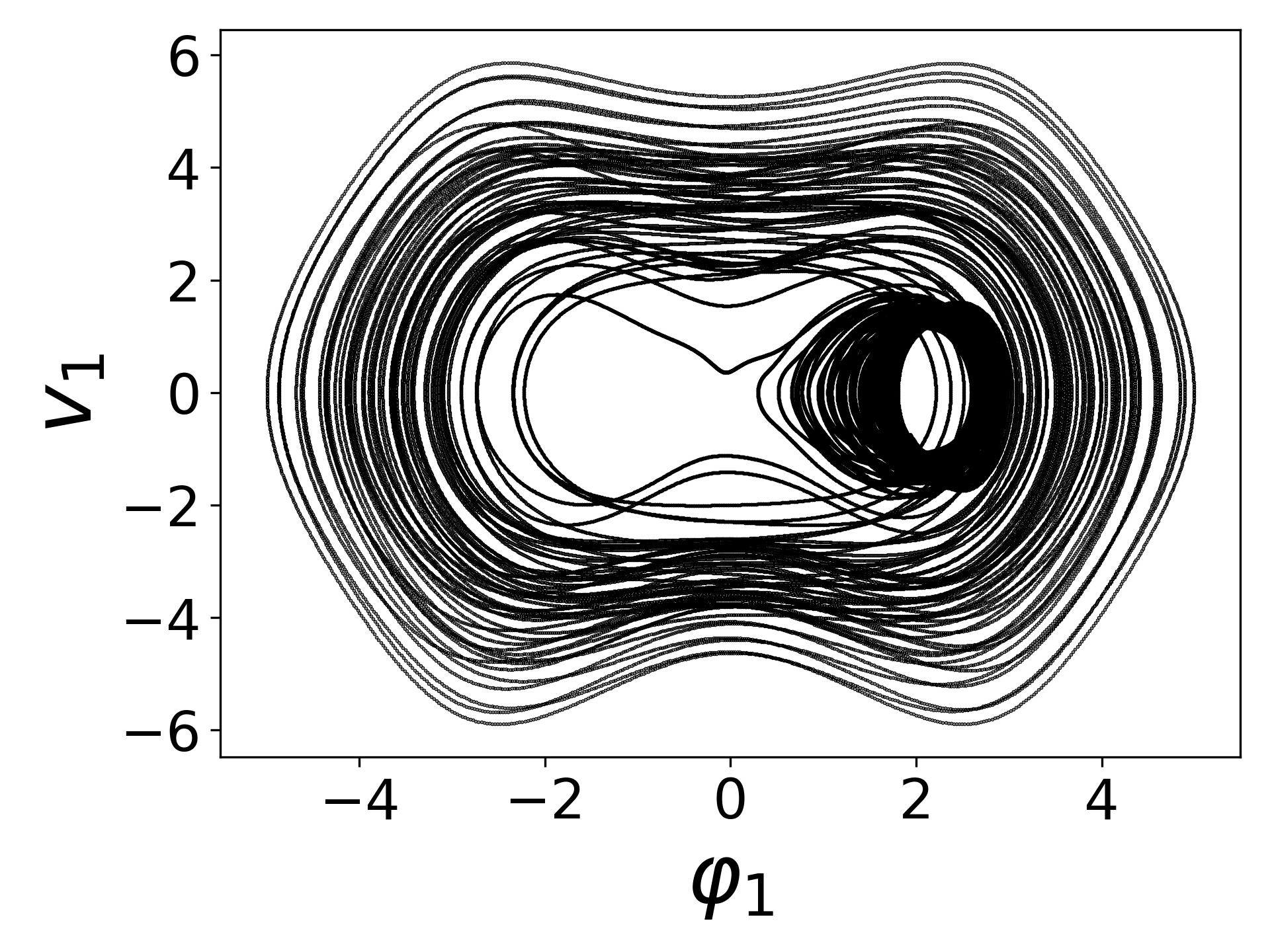}}
\caption{Progressive rising of friction on the pins. Parameters: $m_A = 1$, $m_B = 2$, $m_C = 1$,  $l_1 = 1$, $l_2 = 1$, $l_3 = 1$, $k_1 = 2$, $F = 0$. Initial conditions: $\varphi_1(0) = \pi$, $\varphi_2(0) = 0$, $v_1(0) = v_2(0) = 0$.}
\label{friction_pins_prog_mu}
\end{figure}

\section{A discrete version}\label{sec_discrete}
In this section we consider a discrete version of the equations of motion \eqref{motion4}. For simplicity, we rename the canonical variables as $(\varphi_1, p_1, \varphi_2, p_2) \mapsto (x, y, z, \omega)$ and the parameters as
\begin{subequations}
\beq
A_{11} \mapsto a
\eeq
\beq
A_{12} \mapsto a - \Delta \cos x
\eeq
\beq
A_{22} \mapsto a + b - 2 \Delta \cos x
\eeq
\beq
r_1 \mapsto -k_1 x + \Delta \omega^2 \sin x
\eeq
\beq
r_1 \mapsto -k_2 y - \Delta y (y + 2 \omega) \sin x - c \sin x \,,
\eeq
\end{subequations}
with $a > 0$, $b > 0$; in this way we preserve the physical meaning of the equations. A discrete version of the system \eqref{motion4} can be obtained by simply replacing the derivatives with respect to time with the $(n+1)$-th term in the sequences, that is we define a map
\beq\begin{split}
&f \colon \mathbb{R}^4 \to \mathbb{R}^4 \\
&x_{n+1} = f(x_n), \quad n \in \mathbb{N}
\end{split}\eeq
such that
\begin{subequations}\label{motion4_discrete}
\beq
x_{n+1} = y_n
\eeq
\beq\begin{split}
y_{n+1} =& -(a+b) k_1 x_n + a k_2 z_n +[ (a+b) \Delta \omega_n^2 + a \Delta y_n (y_n + 2 \omega_n) + ac ] \sin x_n + \\
&+ (2 k_1 x_n - k_2 z_n) \Delta \cos x_n - \\
&- [2 \Delta \omega_n^2 + \Delta y_n (y_n + 2 \omega_n) + c] \Delta \sin x_n \cos x_n
\end{split}\eeq
\beq
z_{n+1} = \omega_n
\eeq
\beq\begin{split}
\omega_{n+1} =& a k_1 x_n - a k_2 z_n -[ a \Delta \omega_n^2 + a \Delta y_n (y_n + 2 \omega_n) + ac ] \sin x_n - \\
&- k_1 \Delta x_n \cos x_n + \Delta^2 \omega_n^2 \sin x_n \cos x_n \,,
\end{split}\eeq
\end{subequations}
where all the variables of the system are defined on the entire real set. The constraint $\Delta = 0$, that represents a non-Hamiltonian integrable case for the standard Ziegler pendulum, is again a symmetry that simplifies the system; in this case the equations \eqref{motion4_discrete} become
\begin{subequations}\label{motion4_discrete_delta0}
\beq
x_{n+1} = y_n
\eeq
\beq
y_{n+1} = \frac{1}{ab} [ -(a+b) k_1 x_n + a k_2 z_n + ac \sin x_n ]
\eeq
\beq
z_{n+1} = \omega_n
\eeq
\beq
\omega_{n+1} = \frac{1}{ab} [ a k_1 x_n - a k_2 z_n - ac \sin x_n ]
\eeq
\end{subequations}
and they are well defined $\forall (x, y, z, \omega) \in \mathbb{R}^4$, since we imposed $a > 0, b >0$ from the beginning. We now look for fixed and periodic points of the map \eqref{motion4_discrete_delta0} in the attempt to prove that the discrete map associated to the Ziegler pendulum is chaotic in the sense of Devaney, i.e. it is topologically transitive and has a dense set of periodic points (\cite{Devaney}, \cite{Banks}).

Before going through the calculations, we compute the Jacobian associated to the system \eqref{motion4_discrete_delta0}
\beq
J =
\begin{pmatrix}
0 & 1 & 0 & 0 \\
-\frac{a+b}{ab} k_1 + \frac{c}{b} \cos x_n & 0 & \frac{k_2}{b} & 0 \\
0 & 0 & 0 & 1 \\
\frac{k_1}{b} - \frac{c}{b} \cos x_n & 0 & -\frac{k_2}{b} & 0 \\
\end{pmatrix}
\eeq
that has determinant
\beq
|J| = \frac{k_1 k_2}{ab} \,,
\eeq
so it is contractive if $k_1 k_2 < ab$, conservative if $k_1 k_2 = ab$ and expansive if $k_1 k_2 > ab$. The eigenvalues of $J$ satisfy
\beq\label{eigenvalues}
\lambda^4 + \lambda^2 \bigg( \frac{k_2}{b} + \frac{a+b}{ab} k_1 - \frac{c}{b} \cos x_n \bigg) = \frac{k_2}{b} \bigg( - \frac{a+b}{ab} k_1 + \frac{k_1}{b} \bigg) \,.
\eeq
In the sequel the variable written without index is automatically taken as the $n$-th term of the sequence, that is $x_n =: x$, $y_n =: y$, etc.

\subsection{Fixed points}
We look for the fixed points of the map \eqref{motion4_discrete_delta0}, that is $x_{n+1} = x_n$, etc. A trivial one is $(0, 0, 0, 0)$, while for a generic fixed point we have
\begin{subequations}\label{fixed_eq}
\beq
y = x
\eeq
\beq
-(a+b) k_1 x + a k_2 z + ac \sin x = ab y
\eeq
\beq
\omega = z
\eeq
\beq
a k_1 x - a k_2 z - ac \sin x = ab \omega \,.
\eeq
\end{subequations}
\\
\begin{prop}\label{prop_density}
The set of fixed points of the map \eqref{motion4_discrete_delta0} is not dense in $\mathbb{R}^4$.
\end{prop}
\begin{proof}
The fixed points of the map \eqref{motion4_discrete_delta0} solve the \eqref{fixed_eq}, that is
\begin{subequations}
\beq\label{period1a}
-(a+b) k_1 x + a k_2 z + ac \sin x = ab x
\eeq
\beq\label{period1b}
a k_1 x - a k_2 z - ac \sin x = ab z \,.
\eeq
\end{subequations}
By summing the \eqref{period1a} and the \eqref{period1b} we obtain
\beq
z = - \frac{k_1 + a}{a} x \,,
\eeq
that replaced in the original system returns
\beq\label{xsinx}
\sin x = \frac{[b (k_1 + a) + a k_1 + k_2 (k_1 + a) ]}{ac} x =: \frac{\alpha}{c} x \,,
\eeq
where $\alpha > 0$; so a generic fixed points of the map \eqref{motion4_discrete_delta0} is given by
\beq\label{fixed}
(x_0, y_0, z_0, \omega_0) = \tilde x(\alpha) ( 1, 1,  \beta, \beta ) \,,
\eeq
where $\tilde x(\alpha)$ is one solution of the \eqref{xsinx} for a given value of $\alpha$ and $\beta := -\frac{k_1 + a}{a} < 0$. \\If we restrict ourselves on the interval $x \in [-\pi, \pi]$, the \eqref{xsinx} has the only solution $x = 0$ if $\alpha \ge |c|$, while it admits the solution $x = 0$ and two non null symmetric solutions if $\alpha < |c|$, that is
\beq\label{constraint_c}
b (k_1 + a) + a k_1 + k_2 (k_1 + a) < |c| \,.
\eeq
By extending the domain to $x \in [-n \pi, n \pi]$ with $n \in \mathbb{N}$, we have that for a sufficiently small value of $\alpha$  the \eqref{xsinx} admits an odd number of solutions, more and more large as $\alpha$ is close to zero; indeed $x = 0$ is always a solution of \eqref{xsinx}, while if $x \ne 0$ is a solution, then $-x$ is a solution too, thanks to the odd symmetry of the sine function. Therefore we can write the set of fixed points of the map \eqref{motion4_discrete_delta0} for a certain choice on $\alpha$ as follows
\beq\label{fixed_complete}
U(\alpha) = \bigg\{ \tilde x(\alpha) ( 1, 1, \beta, \beta ) \bigg\}_{\sin(\tilde x(\alpha)) = \alpha \tilde x(\alpha)} \,.
\eeq
Since $\alpha > 0$, \eqref{fixed_complete} is always a finite set of points. For a certain value of $\alpha$, the closure of $U$ is equal to the closed interval that joins the smallest and largest solutions of \eqref{xsinx}; therefore $\overline{U} \ne \mathbb{R}^4$ and $U$ is not dense for any $\alpha > 0$.
\end{proof}
Disregarding for a moment the original domains we imposed for the parameters, we have the following. \\
\begin{cor}
The set of fixed points of the map \eqref{xsinx} is not dense in $\mathbb{R}^4$ for any real value of the parameters.
\end{cor}
\begin{proof}
Proposition \ref{prop_density} is symmetrically generalizable for $\alpha < 0$. For $\alpha = 0$ and remembering \eqref{fixed_complete}, the set of fixed points of the map \eqref{motion4_discrete_delta0} is given by
\beq
U(0) = \bigg\{ 2 k \pi ( 1, 1, \beta, \beta ) \bigg\}_{k \in \mathbb{Z}} \,.
\eeq
The points of $U$ are uniquely defined by the real number $\tilde x$, so it is sufficient to find an open interval of $\mathbb{R}$ that contains $\tilde x$ but does not contain $\beta \tilde x$; furthermore, it is sufficient to find this interval for just one point of $U$. A trivial choice for $k = 0$ is $I = (0, 2 \tilde x)$ if $\tilde x > 0$, $I = (- 2 \tilde x, 0)$ if $\tilde x < 0$ or $I = (-\varepsilon, \varepsilon)$ with $\varepsilon < \pi$ if $\tilde x = 0$; in any case there exists an open set that does not intersect $U$, therefore $U$ is not dense.
\end{proof}
As an example we briefly discuss the stability of the fixed points through the following numerical example. If we put
\beq\begin{split}
&a = b = 1 \\
&k_1 = k_2 = 0.5 \\
&c = 3
\end{split}\eeq
the \eqref{constraint_c} holds and the system \eqref{motion4_discrete_delta0} admits the following fixed points
\beq\begin{split}
&(0, 0, 0, 0); \quad \pm 0.71623935 \bigg( 1, 1, -\frac{3}{2}, -\frac{3}{2} \bigg) \,.
\end{split}\eeq
We can easily establish the stability of these fixed points is verified by substituting in the \eqref{eigenvalues} the previous values for the parameters, to get
\beq
4 \lambda^4 - 12 \lambda^2 \bigg( \cos x_n - \frac{1}{2} \bigg) + 1 = 0 \,.
\eeq
The fixed point associated to $x = 0$ has four real eigenvalues and two of them have magnitude greater than one, so it is a hyperbolic point; the two fixed points associated to $x = \pm 0.71623935$ have instead four complex eigenvalues, so they are not hyperbolic. Notice that the parity of the cosine in the \eqref{eigenvalues} ensures us that the two symmetric fixed points have always the same stability.

We consider the density of the set of the fixed points since, as we will see, it is strongly suggested that the sets of periodic points of the map can be generally reduced to the set \eqref{fixed_complete}, so that from Proposition \ref{prop_density} immediately follows that the map \eqref{motion4_discrete_delta0} is not chaotic in the sense of Devaney for any choice of the parameters.

\subsection{Periodic points of period $2$}
We look for the $2$-periodic points of the map \eqref{motion4_discrete_delta0}, that is $x_{n+2} = x_n$, etc. We have
\begin{subequations}\label{2periodic_eq}
\beq
x_{n+2} = y_{n+1} = \frac{1}{ab} \bigg\{ -(a+b) k_1 x_n + a k_2 z_n +ac \sin x_n \bigg\} = x_n
\eeq
\beq
y_{n+2} = \frac{1}{ab} \bigg\{ -(a+b) k_1 x_{n+1} + a k_2 z_{n+1} +ac \sin x_{n+1} \bigg\} = y_n
\eeq
\beq
z_{n+2} = \omega_{n+1} = \frac{1}{ab} \bigg\{ a k_1 x_n - a k_2 z_n -ac \sin x_n \bigg\} = z_n
\eeq
\beq
\omega_{n+2} = \frac{1}{ab} \bigg\{ a k_1 x_{n+1} - a k_2 z_{n+1} -ac \sin x_{n+1} \bigg\} = \omega_n \,.
\eeq
\end{subequations}
\\
\begin{prop}\label{prop_density2}
The set of $2$-periodic points of the map \eqref{motion4_discrete_delta0} is not dense in $\mathbb{R}^4$.
\end{prop}
\begin{proof}
We substitute the terms $x_{n+1} = y_n$, $z_{n+1} = \omega_n$ in the \eqref{2periodic_eq}, to get
\begin{subequations}
\beq\label{period2a}
-(a+b) k_1 x + a k_2 z + ac \sin x = ab x
\eeq
\beq\label{period2b}
-(a+b) k_1 y + a k_2 \omega + ac \sin y = ab y
\eeq
\beq\label{period2c}
a k_1 x - a k_2 z - ac \sin x = ab z
\eeq
\beq\label{period2d}
a k_1 y - a k_2 \omega - ac \sin y = ab \omega \,.
\eeq
\end{subequations}
By summing the \eqref{period2a} and the \eqref{period2c} and summing the \eqref{period2b} and the \eqref{period2d} we obtain
\begin{subequations}
\beq
z = -\frac{k_1 + a}{a} x
\eeq
\beq
\omega = -\frac{k_1 + a}{a} y \,,
\eeq
\end{subequations}
that is we get exactly the set \eqref{fixed_complete}, that has been shown to be not dense; so the map \eqref{motion4_discrete_delta0} has not a dense set of $2$-periodic points.
\end{proof}

\subsection{Periodic points of period $3$ and conjecture}
We now look for the $3$-periodic points of the map \eqref{motion4_discrete_delta0}, that is $x_{n+3} = x_n$, etc. We have
\begin{subequations}\label{3periodic_eq}
\beq
x_{n+3} = y_{n+2} = \frac{1}{ab} \bigg\{ -(a+b) k_1 x_{n+1} + a k_2 z_{n+1} + ac \sin x_{n+1} \bigg\} = x_n
\eeq
\beq
y_{n+3} = \frac{1}{ab} \bigg\{ -(a+b) k_1 x_{n+2} + a k_2 z_{n+2} +ac \sin x_{n+2} \bigg\} = y_n
\eeq
\beq
z_{n+3} = \omega_{n+2} = \frac{1}{ab} \bigg\{ a k_1 x_{n+1} - a k_2 z_{n+1} -ac \sin x_{n+1} \bigg\} = z_n
\eeq
\beq
\omega_{n+3} = \frac{1}{ab} \bigg\{ a k_1 x_{n+2} - a k_2 z_{n+2} -ac \sin x_{n+2} \bigg\} = \omega_n \,.
\eeq
\end{subequations}
\\
\begin{prop}\label{prop_density3}
The set of $3$-periodic points of the map \eqref{motion4_discrete_delta0} is not dense in $\mathbb{R}^4$.
\end{prop}
\begin{proof}
We substitute the terms $x_{n+1} = y_n$, $z_{n+1} = \omega_n$ and $x_{n+2} = y_{n+1}$, $z_{n+2} = \omega_{n+1}$ in the \eqref{3periodic_eq}, to get
\begin{subequations}
\beq\label{period3a}
-(a+b) k_1 y_n + a k_2 \omega_n + ac \sin y_n = ab x_n
\eeq
\beq\label{period3b}
-(a+b) k_1 y_{n+1} + a k_2 \omega_{n+1} + ac \sin y_{n+1} = ab y_n
\eeq
\beq\label{period3c}
a k_1 y_n - a k_2 \omega_n - ac \sin y_n = ab z_n
\eeq
\beq\label{period3d}
a k_1 y_{n+1} - a k_2 \omega_{n+1} - ac \sin y_{n+1} = ab \omega_n \,.
\eeq
\end{subequations}
By summing and subtracting the \eqref{period3a} and the \eqref{period3c} we obtain
\begin{subequations}\label{period3_constraint}
\beq
-k_1 y = a (x + z)
\eeq
\beq
k_2 \omega = -(a + 2b) z - a x - c \sin \bigg[ -\frac{a}{k_1} (x+z) \bigg] \,.
\eeq
\end{subequations}
We can finally explicit $y_{n+1}$ and $\omega_{n+1}$ in \eqref{period3b} and \eqref{period3d}, substitute the previous constraints and get
\begin{subequations}\label{period3}
\beq
P(x, z) + A \sin x + B \sin \bigg[ P(x, z) + A \sin x \bigg] = 0
\eeq
\beq
Q(x, z) + C \sin x + D \sin \bigg[ Q(x, z) + C \sin x \bigg] = E \sin \bigg[ F (x + z) \bigg] \,,
\eeq
\end{subequations}
where $A, \dots, F$ are constants (that depend on the usual parameters) and $P(x, z)$, $Q(x, z)$ are polynomial functions. \\Regardless the specific form of the system, we have the following set of $3$-periodic points:
\beq
U = \bigg\{ \bigg( \tilde x, \tilde y(\tilde x, \tilde z), \tilde z, \tilde \omega(\tilde x, \tilde z) \bigg) \bigg\} \,,
\eeq
where $(\tilde x, \tilde z)$ are solutions of the \eqref{period3} and $\tilde y$, $\tilde \omega$ are obtained from the \eqref{period3_constraint}, for a certain choice of the parameters. Once we take a proper redefinition of the sine's arguments, we get for the $3$-periodic points an equation analogous to \eqref{xsinx}; therefore we can generalize the Proposition \ref{prop_density} to this case and conclude that the set of $3$-periodic points is not dense in $\mathbb{R}^4$.
\end{proof}
Given the previous results, we are in position to propose the following. \\
\begin{prop}[conjectured]\label{prop_conj}
The map \eqref{motion4_discrete_delta0} has not a dense set of periodic points.
\end{prop}
The validity of the previous conjecture is suggested by the fact that, by increasing the period of the periodic points, we get a series of concatenate sine functions, so that the previous considerations can be extended to an arbitrary period. If the Proposition \ref{prop_conj} is true, the map \eqref{motion4_discrete_delta0} is not chaotic in the sense of Devaney for a choice of parameters  ($k_2 \ne 0$) that has been numerically proven as a general case of chaotic motion for the Ziegler pendulum. Anyway, there is no reason for the regular or chaotic behavior to be preserved by passing from the continuous version to a discrete version of a specific dynamical system, or viceversa. As an example, the logistic map \cite{May} shows the opposite behavior, since the original discrete system shows a very chaotic behavior for a certain choice of the parameter, as one can see from the associated bifurcation diagram, while its continuous version is a simple ordinary equation with simple and regular solutions.

We may ask ourselves whether the discrete map associated to the Ziegler pendulum is chaotic in the sense of Devaney if we add a dissipative force in the equations, for instance the terms associated to the friction on the pins. Assuming as before $\Delta = 0$ and with a proper redefinition of the parameters, the equations for the discrete system become
\begin{subequations}\label{motion4_discrete_delta0_friction}
\beq
x_{n+1} = y_n
\eeq
\beq\begin{split}
y_{n+1} = \frac{1}{ab} [& -(a+b) k_1 x_n + a k_2 z_n + ac \sin x_n - \\
&- a \mu_O \omega \cos (2 z_n) - a \mu_A y \cos (x_n + 2 z_n) ]
\end{split}\eeq
\beq
z_{n+1} = \omega_n
\eeq
\beq\begin{split}
\omega_{n+1} = \frac{1}{ab} [& a k_1 x_n - a k_2 z_n - ac \sin x_n + a \mu_O \omega \cos (2 z_n) + a \mu_A y \cos (x_n + 2 z_n) ] \,,
\end{split}\eeq
\end{subequations}
that is we modify the original equations \eqref{motion4_discrete_delta0} by adding the same term in the sequence defining $y_n$ (with a negative sign) and in the sequence defining $\omega_n$ (with a positive sign). It is not difficult to see that this symmetry in the equations allows us to generalize all the previous considerations. Given that, regardless an explicit study of this case, we can conjecture that also the map \eqref{motion4_discrete_delta0_friction} is not chaotic in the sense of Devaney, since it has the same identical set of fixed points and similar sets of periodic points when compared with the map \eqref{motion4_discrete_delta0}.

\section{Conclusion}\label{sec_conclusion}
In this paper we studied the dynamics of a generalized Ziegler pendulum, when subject to further external forces, both potential and non-potential.

We found that, in general, the presence of gravity destroys the integrability of the system that is present in the original dynamical system if two specific symmetries on the parameters hold, even if isolated periodic orbits can be found for certain choices of parameters and initial conditions. By adding instead a further elastic potential energy that preserves the cyclicity of one of the two generalized coordinates, the integrable cases of the original system survive, so that closed trajectories and families of periodic solutions arise. The physical version of the system does not suggest any relevant considerations, that is the regular or chaotic behavior of the system is essentially the same as if we consider a mathematical, physical or mixed double pendulum.

An interesting behavior is observed if the system is subject to a dissipative force. In the case of fluid-like friction we found a specific symmetry on the friction coefficients under which the system admits attractive points, while in general it approaches to limit cycles. Another formulation of dissipative force has been studied, by considering torsional friction forces on the pendulum pins. By numerically studying a progressive breaking of the non-Hamiltonian symmetry $\Delta = 0$, we observe the emergence of a possible strange attractor, that seems to arise in a different way if compared with the strange attractors such as the Lorenz attractor. Furthermore, a brief study of the transition to chaos has been performed, with the aim of finding a threshold value for the friction coefficients able to distinguish between an approximately regular motion and chaotic orbits.

Finally, we found that the discrete map associated to the Ziegler pendulum, if the non-Hamiltonian symmetry $\Delta = 0$ holds, does not have dense sets of periodic points up to period $3$. A qualitative analysis of the periodic points suggests that the map, with this constraint, is not chaotic in the sense of Devaney, for a choice of parameters that corresponds generally to chaotic orbits for the continuous Ziegler pendulum.

The original Ziegler pendulum introduced in \cite{Ziegler} has found several applications in the last decades. The variants analyzed in this work, in particular the one that takes into account the presence of friction on the pins, may be considered as more realistic models in order to improve the mechanical and engineering applications of this dynamical system; furthermore, the occurrence of a limit cycle in presence of fluid-like friction may be taken into account for further studies on optimization and control of the system. Finally, the results here presented relating the discrete map associated to the dynamical system may be developed to further analyze the definition of chaos in the sense of Devaney.

In view of possible future developments several questions can now be asked.
\begin{itemize}
\item Given the nature of the sets of periodic points, it should be possible, even if not necessary easy, to prove the Proposition \ref{prop_conj} by induction on the period; the main difficulty is related to the quite complex form associated to the equations solved by the $k$-periodic points for periods $k > 2$.
\item For all the variants considered, it is evident that the non-Hamiltonian symmetry $\Delta = 0$ plays a fundamental role in the distinction between regular and chaotic motion, but the reason is not trivial.
\item We may wonder why we have a periodic orbit in presence of gravity for the choices of parameters shown in Figure \ref{gravity_periodic_compare}; for these considerations a deep study of the applications of KAM theory for this specific dynamical system could be considered.
\item The presence of dissipative forces on the pins seems to give rise to a strange attractor, at least for a specific choice of parameters; in order to state this in a rigorous way, an extensive and detailed numerical study for several values of parameters and initial conditions must be performed.
\item Since the Hamilton equations of the standard Ziegler pendulum present a periodic generalized force, the Poincaré-Mel'nikov method \cite{Poincaré, Melnikov} can be applied to the system, in order to verify in different fashion the chaotic behavior of the system with the breaking of the mentioned symmetries.
\end{itemize}

\backmatter

\bmhead{Acknowledgements} The authors are grateful to the anonymous reviewer for his useful suggestions.

\section*{Declarations}

\bmhead{Funding} FIRD 2022 contract - University of Ferrara.

\bmhead{Conflict of interest/Competing interests} The authors have no conflict or competing interests to disclose.

\bmhead{Ethics approval and consent to participate} Not applicable.

\bmhead{Consent for publication} Not applicable.

\bmhead{Data availability} The data that supports the fundings of this study are available within the article.

\bmhead{Materials availability} Not applicable.

\bmhead{Code availability} Not applicable.

\bmhead{Author contribution} The authors contributed in equal parts to the paper.






\bibliography{sn-bibliography}

@article{Ziegler,
  author		= {Ziegler, H.},
  title			= {Die stabilitätskriterien der elastomechanik},
  journal		= {Ingenieur-Archiv},
  volume		= {20},
  number		= {1},
  pages		= {49--56},
  year			= {1952},
}

@book{Pfluger,
  author		= {Pflüger, A.},
  title			= {Stabilitätsprobleme der Elastostatik},
  address		= {Berlin},
  publisher		= {Springer Verlag},
  year			= {1952},
}

@article{Shinbrot,
  author		= {Shinbrot, T. and Grebogi, C. and Wisdom, J. and Yorke, J. A.},
  title			= {Chaos in a double pendulum},
  journal		= {American Journal of Physics},
  volume		= {60},
  number		= {6},
  pages		= {491--499},
  year			= {1992},
}

@article{Stachowiak,
  author		= {Stachowiak, T. and Okada, T.},
  title			= {A numerical analysis of chaos in the double pendulum},
  journal		= {Chaos, Solitons \& Fractals},
  volume		= {29},
  number		= {2},
  pages		= {417--422},
  year			= {2006},
}

@article{Dullin,
  author		= {Dullin, H. R.},
  title			= {Melnikov’s method applied to the double pendulum},
  journal		= {Zeitschrift für Physik B Condensed Matter},
  volume		= {93},
  number		= {4},
  pages		= {521--528},
  year			= {1994},
}

@article{Polekhin,
  author		= {Polekhin, I. Yu.},
  title			= {On the dynamics and integrability of the {Z}iegler pendulum},
  journal		= {Nonlinear Dynamics},
  year			= {2024},
  doi			= {10.1007/s11071-024-09444-8},
}

@article{Kozlov,
  author		= {Kozlov, V. V.},
  title			= {On the integrability of circulatory systems},
  journal		= {Regular and Chaotic Dynamics},
  volume		= {27},
  number		= {1},
  pages		= {11--17},
  year			= {2022},
}

@article{Thomsen,
  author		= {Thomsen, J. J.},
  title			= {Chaotic dynamics of the partially follower-loaded elastic double pendulum},
  journal		= {Journal of Sound and Vibration},
  volume		= {188},
  number		= {3},
  pages		= {385--405},
  year			= {1995},
}

@article{Bigoni2011,
  author		= {Bigoni, D. and Noselli, G.},
  title			= {Experimental evidence of flutter and divergence instabilities induced by dry friction},
  journal		= {Journal of the Mechanics and Physics of Solids},
  volume		= {59},
  pages		= {2208--2226},
  year			= {2011},
}

@inproceedings{Kirillov2022,
  author		= {Kirillov, O. N. and Verhulst, F.},
  title			= {From rotating fluid masses and {Z}iegler’s paradox to {P}ontryagin- and {K}rein spaces and bifurcation theory},
  editor		= {Günther, M. and Schilders, W.},
  volume		= {38},
  booktitle		= {Novel {M}athematics {I}nspired by {I}ndustrial {C}hallenges. Mathematics in {I}ndustry},
  pages		= {201--243},
  address		= {Cham},
  publisher		= {Springer},
  year			= {2022},
}

@misc{Bigoni2024,
  author		= {Bigoni, D. and Dal Corso, F. and Kirillov, O. N. and Misseroni, D. and Noselli, G. and Piccolroaz, A.},
  title			= {Flutter instability in solids and structures, with a view on biomechanics and metamaterials},
  year			= {2024},
  note			= {Preprint at \url{https://arxiv.org/abs/2401.07092v1}},
}

@article{Kirillov2011,
  author		= {Kirillov, O. N.},
  title			= {Singularities in {S}tructural {O}ptimization of the {Z}iegler {P}endulum},
  journal		= {Acta Polytechnica},
  volume		= {51},
  number		= {4},
  year			= {2011},
  doi			= {10.14311/1400},
}

@book{Arnold,
  author		= {Arnold, V. I. and Kozlov, V. V. and Neishtadt, A. I.},
  title			= {Mathematical aspects of classical and celestial mechanics, volume 3},
  address		= {Berlin},
  publisher		= {Springer},
  year			= {2007},
}

@book{Jacobi,
  author		= {Jacobi, C. G. J. and Borchardt, C. W. and Clebsch, A. and Lottner, E.},
  title			= {CGJ Jacobi’s Vorlesungen über Dynamik},
  address 		= {Berlin},
  publisher		= {G. Reimer},
  year			= {1884},
}

@book{Devaney,
  author		= {Devaney, R. L.},
  title			= {An Introduction to Chaotic Dynamical Systems, Second Edition},
  address 		= {Redwood City},
  publisher		= {Addison-Wesley},
  year			= {1989},
}

@article{Banks,
  author		= {Banks, J. and Brooks, J. and Cairns, G. and Davis, G. and Stacey, P.},
  title			= {On {D}evaney’s {D}efinition of {C}haos},
  journal		= {The American Mathematical Monthly},
  volume		= {99},
  number		= {4},
  pages		= {332--334},
  year			= {1992\color{black}},
}

@article{Wolf,
  author		= {Wolf, A. and Swift, J. B. and Swinney, H. L. and Vastano, J. A.},
  title			= {Determining {L}yapunov {E}xponents from a {T}ime {S}eries},
  journal		= {Physica D},
  volume		= {16},
  number		= {3},
  pages		= {285--317},
  year			= {1985},
}

@misc{Yanchuk, 
  author		= {Yanchuk, S. and Perlikowski, P. and Wolfrum, M. and Stefański, A. and Kapitaniak, T.},
  title			= {Fast transition to chaos in a ring of undirectionally coupled oscillators},
  year			= {2011},
  note			= {Preprint at \url{https://opus4.kobv.de/opus4-matheon/frontdoor/index/index/docId/820}},
}

@article{Lorenz,
  author		= {Lorenz, E. N.},
  title			= {Deterministic {N}onperiodic {F}low},
  journal		= {Journal of the Atmospheric Sciences},
  volume		= {20},
  number		= {2},
  pages		= {130--141},
  year			= {1963},
}

@article{May,
  author		= {May, R. M.},
  title			= {Simple mathematical models with very complicated dynamics},
  journal		= {Nature},
  volume		= {261},
  pages		= {459--467},
  year			= {1976},
}

@article{Poincaré,
  author		= {Poincaré, H.},
  title			= {Sur le problème des trois corps et les équations de la dynamique},
  journal		= {Acta Mathematica},
  volume		= {13},
  pages		= {1--270},
  year			= {1890},
}

@article{Melnikov,
  author		= {Mel'nikov, V. K.},
  title			= {On the stability of a center for time-periodic perturbations},
  journal		= {Tr. Mosk. Mat. Obs.},
  volume		= {12},
  pages		= {3--52},
  year			= {1963},
}

\end{document}